\pdfoutput=1
\documentclass[conference]{IEEEtran}
\makeatletter
\def\ps@headings{%
\def\@oddhead{\mbox{}\scriptsize\rightmark \hfil \thepage}%
\def\@evenhead{\scriptsize\thepage \hfil \leftmark\mbox{}}%
\def\@oddfoot{}%
\def\@evenfoot{}}
\makeatother
\usepackage[pdftex]{graphicx}
\usepackage{amsmath}
\usepackage{amssymb}
\usepackage{array}
\usepackage{mdwmath}
\usepackage{subfigure}
\usepackage{hyperref}
\usepackage{algorithm}
\usepackage{algorithmic}
\usepackage{epsfig}
\usepackage{epstopdf}

\graphicspath{{../pdf/}{../jpeg/}}
\DeclareGraphicsExtensions{.pdf,.jpeg,.png}

\newtheorem{theorem}{Theorem}
\newtheorem{lemma}{Lemma}

\newcommand{\beqa}{\begin{eqnarray}}
\newcommand{\eeqa}{\end{eqnarray}}

\floatname{algorithm}{Algorithm}
\begin{document}
\title{Optimal Sequential Wireless Relay Placement \\
on a Random Lattice Path}

\author{
\IEEEauthorblockN{Abhishek Sinha\IEEEauthorrefmark{1}, Arpan Chattopadhyay\IEEEauthorrefmark{1}, K.~P.~Naveen\IEEEauthorrefmark{1}, 
Marceau Coupechoux\IEEEauthorrefmark{2}  and Anurag Kumar\IEEEauthorrefmark{1}}
\IEEEauthorblockA{\IEEEauthorrefmark{1}Dept. of Electrical Communication Engineering,
Indian Institute of Science, Bangalore 560012, India.\\
Email: \{abhishek.sinha.iisc, arpanc.ju\}@gmail.com, \{naveenkp, anurag\}@ece.iisc.ernet.in}
\IEEEauthorblockA{\IEEEauthorrefmark{2}Telecom ParisTech and CNRS LTCI, 
Dept.\ of Informatique et R\'eseaux, 23, avenue d'Italie, 75013 Paris, France.\\
Email: marceau.coupechoux@telecom-paristech.fr}
}
\maketitle

%%%%%%%%%%%%%%%%%%%%%%%%%%%%%%%%%%%%%%%%%%%%%%%%%%%%%%%%%%%%%%%%%%%%%%%%%%%%%%%%%%%%%%%%%
\begin{abstract}
  Our work is motivated by the need for impromptu (or ``as-you-go'')
  deployment of relay nodes (for establishing a packet communication path
  with a control centre) by firemen/commandos while operating in an
  unknown environment. We consider a model, where a deployment
  operative steps along a random lattice path whose evolution is
  Markov. At each step, the path can randomly either continue in the
  same direction or take a turn ``North'' or ``East,'' or come to an
  end, at which point a data source (e.g., a temperature sensor) has
  to be placed that will send packets to a control centre at the
  origin of the path. A decision has to be made at each step whether
  or not to place a wireless relay node. Assuming that the packet
  generation rate by the source is very low, and simple link-by-link
  scheduling, we consider the problem of relay placement so as to
  minimize the expectation of an end-to-end cost metric (a linear
  combination of the sum of convex hop costs and the number of relays
  placed). This impromptu relay placement problem is formulated as a
  total cost Markov decision process.  First, we derive the optimal
  policy in terms of an optimal placement set and show that this set
  is characterized by a boundary beyond which it is optimal to
  place. Next, based on a simpler alternative one-step-look-ahead
  characterization of the optimal policy, we propose an algorithm
  which is proved to converge to the optimal placement set in a finite
  number of steps and which is faster than the traditional value
  iteration. We show by simulations that the distance based heuristic,
  usually assumed in the literature, is close to the optimal provided
  that the threshold distance is carefully chosen.
\end{abstract}

%%%%%%%%%%%%%%%%%%%%%%%%%%%%%%%%%%%%%%%%%%%%%%%%%%%%%%%%%%%%%%%%%%%%%%%%%%%%%%%%%%%%%%%%%%%


\section{Introduction} \label{intro} Wireless networks, such as
cellular networks or multihop ad hoc networks, would normally be
deployed via a planning and design process. There are situations,
however, that require the impromptu (or ``as-you-go'') deployment of a
multihop wireless packet network. For example, such an impromptu
approach would be required to deploy a wireless sensor network for
situational awareness in emergency situations such as those faced by
firemen or commandos (see
\cite{Fischer,howard-etal02incremental-self-deployment-algorithm}). For
example, as they attack a fire in a building, firemen might wish to
place temperature sensors on fire-doors to monitor the spread of fire,
and ensure a route for their own retreat; or commandos attempting to
flush out terrorists might wish to place acoustic or passive infra-red
sensors to monitor the movement of people in the building. As-you-go
deployment may also be of interest when deploying a multi-hop wireless
sensor network over a large terrain (such as a dense forest) in order
to obtain a first-cut deployment which could then be augmented to a
network with desired properties (connectivity and quality-of-service).

With the above larger motivation in mind, in this paper we are
concerned with the rigorous formulation and solution of a problem of
impromptu deployment of a multihop wireless network along a random
lattice path, see Fig.~\ref{impromptu_figure}. The path could
represent the corridor of a large building, or even a trail in a
forest. The objective is to create a multihop wireless path for packet
communication from the end of the path to its beginning. The problem
is formulated as an optimal sequential decision problem. The
formulation gives rise to a total cost Markov decision process, which
we study in detail in order to derive structural properties of the
optimal policy. We also provide an efficient algorithm for calculating
the optimal policy.

\begin{figure}[t]
\centering
\includegraphics[scale=0.45,angle=0]{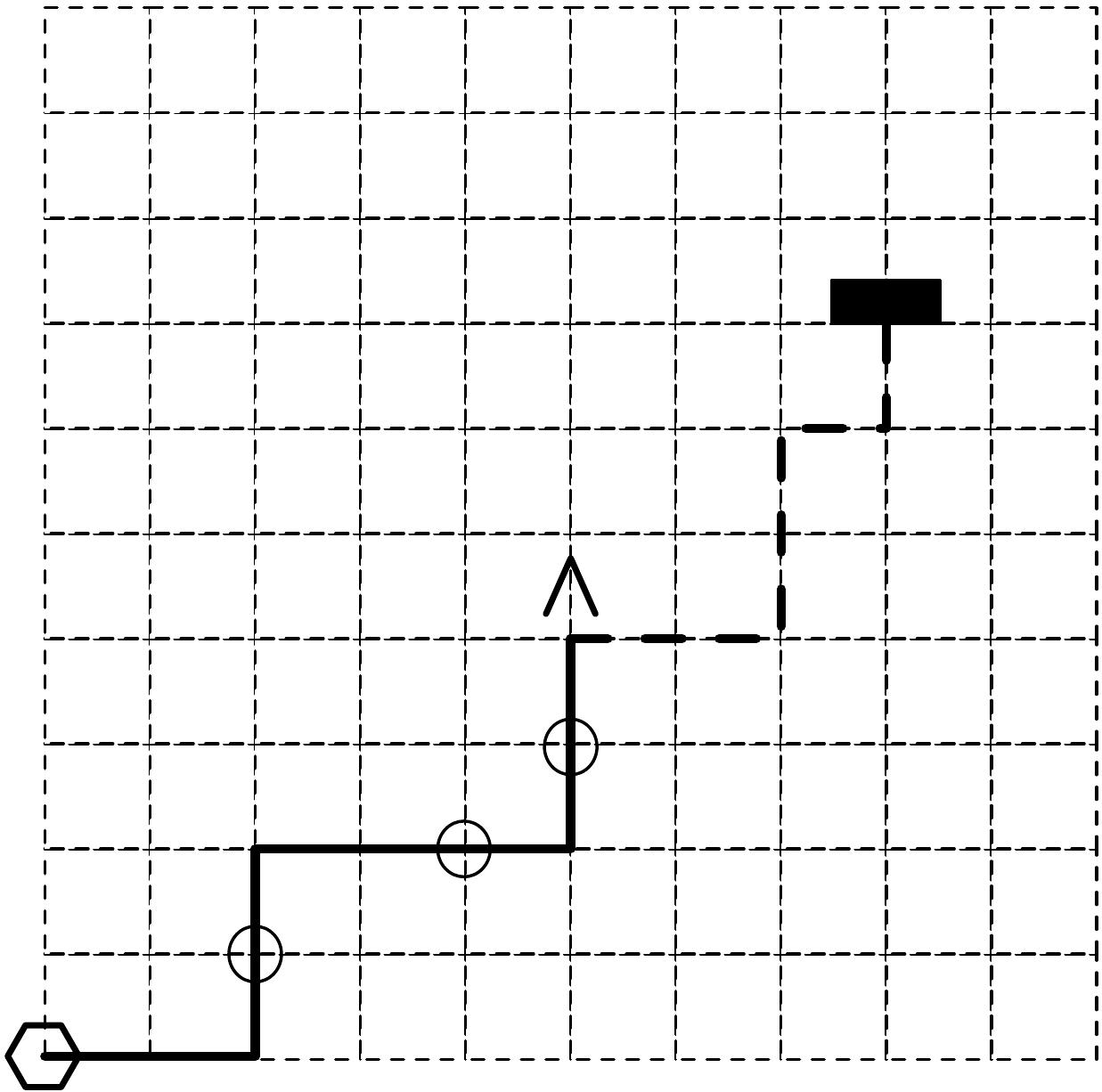}
\caption{\label{impromptu_figure} A wireless network being deployed as
  a person steps along a random lattice path. Inverted \textbf{V}:
  location of the deployment person; solid line: path already covered;
  circles: deployed relays; thick dashed path: a possible evolution of
  the remaining path. The sensor to be placed at the end is
  also shown as the black rectangle.}
\vspace{-6mm}
\end{figure}

\subsection{Related Work}
Our study is motivated by ``first responder'' networks, a concept that
has been around at least since 2001.  In
\cite{howard-etal02incremental-self-deployment-algorithm}, Howard et
al.\ provide heuristic algorithms for the problem of incremental
deployment of sensors (such as surveillance cameras) with the
objective of covering the deployment area. Their problem is related to
that of self-deployment of autonomous robot teams and to the
art-gallery problem.  Creation of a communication network that is
optimal in some sense is not an objective in
\cite{howard-etal02incremental-self-deployment-algorithm}. In a
somewhat similar vein, the work of Loukas et al.\
\cite{mobihoc.loukas-etal08robotic-wireless-network-emergency} is
concerned with the dynamic locationing of robots that, in an emergency
situation, can serve as wireless relays between the infrastructure and
human-carried wireless devices.  The problem of impromptu deployment
of static wireless networks has been considered in
\cite{mobihoc.naudts-etal07monitoring-planning-tool,
  mobihoc.souryal-etal07real-time-deployment-range-extension,
  mobihoc.souryal-etal09rapidly-deployable-mesh-network-testbed,
  mobihoc.aurisch-tlle09relay-placement-emergency-response}.  In
\cite{mobihoc.naudts-etal07monitoring-planning-tool}, Naudts et al.\
provide a methodology in which, after a node is deployed, the next
node to be deployed is turned on and begins to measure the signal
strength to the last deployed node.  When the signal strength drops
below a predetermined level, the next node is deployed and so
on. Souryal et al.\ provide a similar approach in
\cite{mobihoc.souryal-etal07real-time-deployment-range-extension,
  mobihoc.souryal-etal09rapidly-deployable-mesh-network-testbed},
where an extensive study of indoor RF link quality variation is
provided, and a system is developed and demonstrated.  The work
reported in
\cite{mobihoc.aurisch-tlle09relay-placement-emergency-response} is yet
another example of the same approach for relay deployment. More
recently, Liu et al.\ \cite{Breadcrumb} describe a ``breadcrumbs''
system for aiding firefighters inside buildings, and is similar to our
present paper in terms of the class of problems it addresses.  In a
survey article \cite{Fischer}, Fischer et al.\ describe various
localization technologies for assisting emergency responders, thus
further motivating the class of problems we consider.

In our earlier work (Mondal et al.\
\cite{mondal-etal12impromptu-deployment_NCC}) we took the first steps
towards rigorously formulating and addressing the problem of impromptu
optimal deployment of a multihop wireless network on a line.  The line
is of unknown length but prior information is available about its
probability distribution; at each step, the line can come to an end
with probability $p$, at which point a sensor has to be placed. Once
placed, the sensor sends periodic measurement packets to a control
centre near the start of the line.  It is assumed that the measurement
rate at the sensor is low, so that (with a very high probability) a
packet is delivered to the control centre before the next packet is
generated at the sensor. This so called ``lone packet model'' is
realistic for situations in which the sensor makes a measurement every
few seconds. 

The objective of the sequential decision problem is to minimise a
certain expected per packet cost (e.g., end-to-end delay or total
energy expended by a node), which can be expressed as the sum of the
costs over each hop, subject to a constraint on the number of relays
used for the operation. It has been proved in
\cite{mondal-etal12impromptu-deployment_NCC} that an optimal placement
policy solving the above mentioned problem is a threshold rule, i.e.,
there is a threshold $r^*$ such that, after placing a relay, if the
operative has walked $r^*$ steps without the path ending, then a relay
must be placed at $r^*$.

\subsection{Outline and Our Contributions}
In this paper, while continuing to assume (\textbf{a}) that a single
operative moves step-by-step along a path, deciding to place or to not
place a relay, (\textbf{b}) that the length of the path is a
geometrically distributed random multiple of the step size,
(\textbf{c}) that a source of packets is placed at the end of the
path, (\textbf{d}) that the lone packet traffic model applies, and
(\textbf{e}) that the total cost of a deployment is a linear
combination of the sum of convex hop costs and the number of nodes
placed, we extend the work presented in
\cite{mondal-etal12impromptu-deployment_NCC} to the two-dimensional
case. At each step, the line can take a right angle turn either to the
``East'' or to the ``North'' with known fixed probabilities. We assume
a Non-Line-Of-Sight (NLOS) propagation model, where a radio link
exists between two nodes placed anywhere on the path, see
Fig.~\ref{lattice-path-figure}.  The lone packet model is a natural
first assumption, and would be useful in low-duty cycle monitoring
applications.  Once the network has been deployed, an analytical
technique such as that presented in
\cite{rachit-kumar12performance-analysis} can be used to estimate the
actual packet carrying capacity of the network.

We will formally describe our system model and problem formulation in
Section~\ref{system_model_section}.  The following are our main
contributions:

\begin{figure}[t]
\centering
\includegraphics[width=60mm,height=40mm]{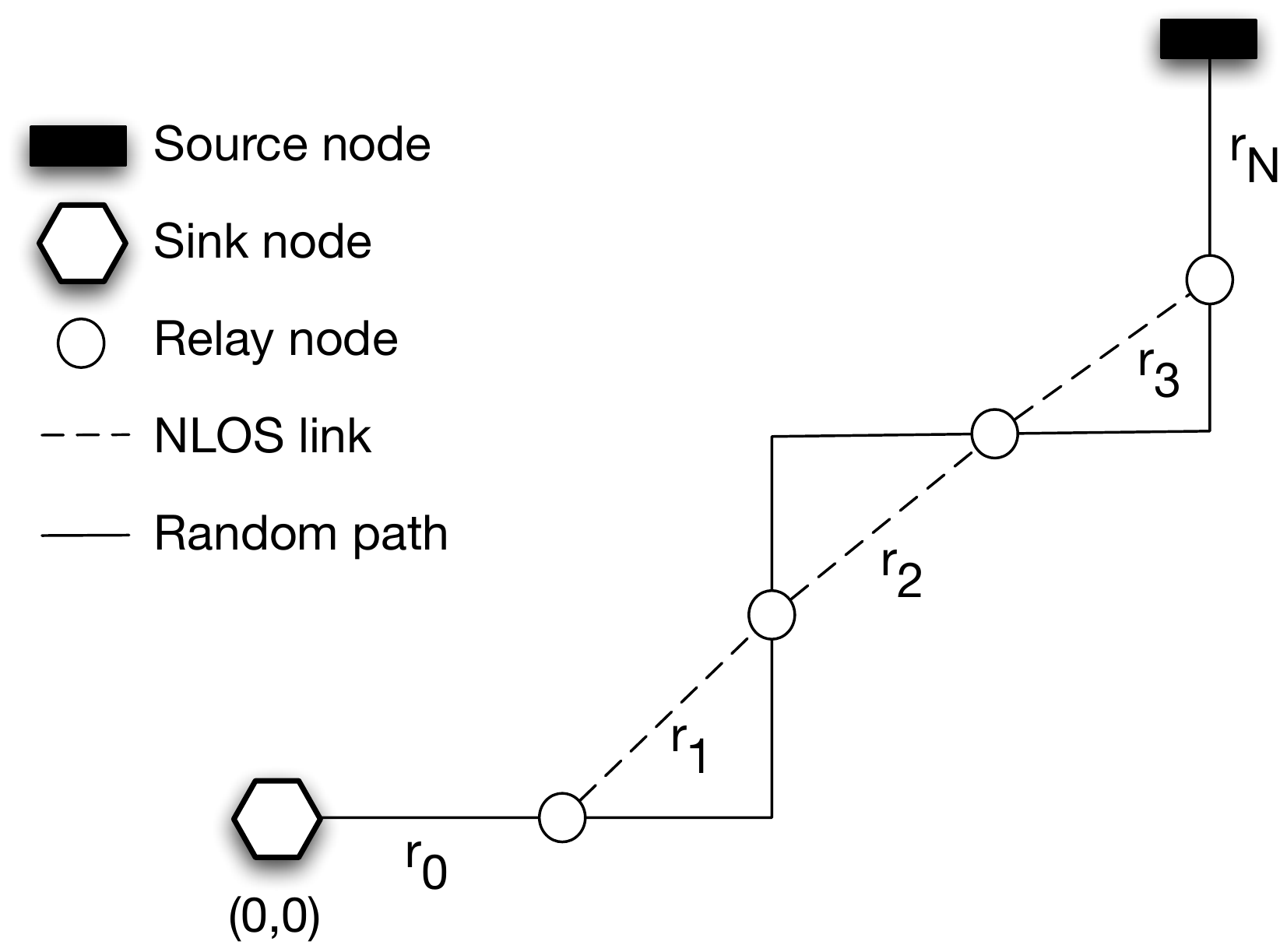}
\caption{\label{lattice-path-figure} A depiction of relay deployment 
along a random lattice path with NLOS propagation.}
\vspace{-6mm}
\end{figure}

\begin{itemize}
\item We formulate the problem as a total cost Markov decision process
  (MDP), and characterize the optimal policies in terms of placement
  sets. We show that these optimal policies are threshold policies and
  thus the placement sets are characterized by boundaries in the
  two-dimensional lattice (Section~\ref{nlos_section}). Beyond these
  boundaries, it is optimal to place a relay.
\item Noticing that placement instants are renewal points in the
  random process, we recognize and prove the One-Step-Look-Ahead
  (OSLA) characterization of the placement sets
  (Section~\ref{OSLA_formulation_section}).
\item Based on the OSLA characterization, we propose an iterative
  algorithm, which converges to the optimal placement set in a finite
  number of steps (Section~\ref{FPI_NLOS_section}). We have observed
  that this algorithm converges much faster than value iteration.
\item In Section~\ref{numerical_work_section} we provide several
  numerical results that illustrate the theoretical development. The
  relay placement approach proposed in
  \cite{mobihoc.naudts-etal07monitoring-planning-tool,
    mobihoc.souryal-etal07real-time-deployment-range-extension,
    mobihoc.souryal-etal09rapidly-deployable-mesh-network-testbed,
    mobihoc.aurisch-tlle09relay-placement-emergency-response} would
  suggest a distance threshold based placement rule. We numerically
  obtain the optimal rule in this class, and find that the cost of
  this policy is numerically indistinguishable from that of the
  overall optimal policy provided by our theoretical development. It
  suggests that it might suffice to utilize a distance threshold
  policy. However, the distance threshold should be carefully designed
  taking into account the system parameters and the optimality
  objective.
\end{itemize}
For the ease of presentation we have moved most of the proofs
to the Appendix. 

\section{System Model} \label{system_model_section} We consider a
deployment person, whose stride length is 1 unit, moving along a
random path in the two-dimensional lattice, placing relays at some of
the lattice points of the path and finally a source node at the end of
the path. Once placed, the source node periodically generates
measurement packets which are forwarded by the successive relays in a
multihop fashion to the control centre located at $(0,0)$; see
Fig.~\ref{lattice-path-figure}.

\subsection{Random Path}

Let $\mathbb{Z}_+$ denote the set of nonnegative integers, and
$\mathbb{Z}_+^2$ the nonnegative orthant of the two dimensional
integer lattice. We will refer to the $\mathsf{x}$ direction as East
and to the $\mathsf{y}$ direction as North. Starting from $(0,0)$
there is a lattice path that takes random turns to the North or to the
East (this is to avoid the path folding back onto itself, see
Fig~\ref{lattice-path-figure}). Under this restriction, the path
evolves as a stochastic process over $\mathbb{Z}_+^2$. When the
deployment person has reached some lattice point, the path continues
for one more step and terminates with probability $p$, or does not
terminate with probability $1-p$. In either case, the next step is
Eastward with probability $q$ and Northward with probability $1-q$.
Thus, for instance, $(1-p)q$ is the probability that the path proceeds
Eastwards without ending. The person deploying the relays is assumed
to keep a count of $m$ and $n$, the number of steps taken in the
$\mathsf{x}$ direction and in $\mathsf{y}$ direction, repectively,
since the previous relay was placed. He is also assumed to know the
probabilities $p$ and $q$.

\subsection{Cost Definition}

In our model, we assume NLOS propagation, i.e., packet transmission can take place between any two 
successive relays even if they are not on the same straight line segment of the lattice path. In the 
building context, this would correspond to the walls being radio transparent. The model is also suitable 
when the deployment region is a thickly wooded forest where the deployment person is restricted to move 
only along some narrow path (lattice edges in our model). 

For two successive relays separated by a distance $r$, we assign a cost of $d(r)$ 
which could be the average delay incurred over that hop (including transmission overheads and 
retransmission delays), or the power required to get a packet across the hop.
For instance, in our numerical work we use the power cost, 
$d(r)=P_m+\gamma r^\eta$, where $P_m$ is the minimum power required, $\gamma$ 
represents an SNR constraint and $\eta$ is the path-loss exponent. Now suppose 
$N$ relays are placed such that the successive 
inter-relay distances are $r_0,r_1,\cdots,r_N$ ($r_0$ is the distance from the 
control centre at $(0,0)$ and the first relay, and $r_N$ is the distance from 
the last relay to the sensor placed at the end of the path) then the total cost 
of this placement is the sum of the one-hop costs $C=\sum_{i=0}^N d(r_i)$.
The total cost being the sum of one-hop costs can be justified for the lone packet 
model since when a packet is being forwarded there is no other packet
transmission taking place. 

We now impose a few technical conditions on the one-hop cost function
$d(\cdot)$: (\textbf{C1}) $d(0)>0$, (\textbf{C2}) $d(r)$ is convex and
increasing in $r$, and (\textbf{C3}) for any $r$ and $\delta>0$ the
difference $d(r+\delta)-d(r)$ increases to $\infty$.

(\textbf{C1}) is imposed considering the fact that it requires a
non-zero amount of delay or power for transmitting a packet between
two nodes, however close they may be. (\textbf{C2}) and (\textbf{C3})
are properties we require to establish our results on the optimal
policies. They are satisfied by the power cost, $P_m+\gamma r^\eta$ ,
and also by the mean hop delay (see \cite{Prasenjit}).

We will overload the notation $d(\cdot)$ by denoting the one-hop cost between 
the locations $(0,0)$ and $(x,y)\in\Re^2$ as simply $d(x,y)$ instead 
of $d(||(x,y)-(0,0)||)$. Using the condition on $d(r)$ we prove the following 
convexity result of $d(x,y)$.
\begin{lemma} \label{conv} The function $d(x,y)$ is convex in $(x,y)$,
  where $(x,y)\in\mathbb{R}^2$.
\end{lemma}
\begin{proof}
  This follows from the fact that $d(\cdot)$ is convex, non-decreasing
  in its argument. For a formal proof, see
  Appendix~\ref{conv_appendix}.
\end{proof}
We further impose the following condition on $d(x,y)$ where
$(x,y)\in\Re^2$.  We allow a general cost-function $d(x,y)$ endowed
with the following property: (\textbf{C4}) The function $d(x,y)$ is
positive, twice continuously partially differentiable in variables $x$
and $y$ and $\forall x,y\in\mathbb{R}_+$,
\begin{eqnarray} \label{assumption}
d_{xx}(x,y)>0,\hspace{5pt}d_{xy}(x,y)>0,\hspace{5pt}d_{yy}(x,y)>0,
\end{eqnarray}
where $d_{xy}(x,y)=\frac{\partial^2d(x,y)}{\partial x \partial
  y}$. These properties also hold for the mean delay and the power
functions mentioned earlier.

Finally define, for $(m,n)\in\mathbb{Z}_+^2$, $\Delta_1(m,n)=d(m+1,n)-d(m,n)$ and 
$\Delta_2(m,n)=d(m,n+1)-d(m,n)$. 
\begin{lemma}
\label{cor1}
$\Delta_1(m,n)$ and $\Delta_2(m,n)$ are 
non-decreasing in both the coordinates $m$ and $n$. 
\end{lemma}
\begin{proof}
  This follows directly from (\ref{assumption}). See
  Appendix~\ref{cor1_appendix} for details.
\end{proof}

\subsection{Deployment Policies and Problem Formulation}
A deployment policy $\pi$ is a sequence of mappings $(\mu_k:k\ge0)$,
where at the $k$-th step of the path (provided that the path has not
ended thus far) $\mu_k$ allows the deployment person to decide whether
to \emph{place} or \emph{not to place} a relay where, in general,
randomization over these two actions is allowed. The decision is based
on the entire information available to the deployment person at the
$k$-th step, namely the set of vertices traced by the path and the
location of the previous vertices where relays were placed.  Let $\Pi$
represent the set of all policies. For a given policy $\pi \in \Pi$,
let $\mathbb{E}_\pi$ represent the expectation operator under policy
$\pi$. Let $C$ denote the total cost incurred and $N$ the total number
of relays used. We are interested in solving the following problem,
\begin{eqnarray}
\label{eq:modified} 
\min_{\pi \in \Pi} & \mathbb{E}_{\pi}C+\lambda \mathbb{E}_{\pi}N,
\end{eqnarray}
where $\lambda>0$ may be interpreted as the cost of a relay.  Solving
the problem in (\ref{eq:modified}) can also help us solve the
following constrained problem,
\begin{eqnarray}
\label{eq:main}
\min_{\pi \in \Pi} & \mathbb{E}_{\pi}C \nonumber \\
\mbox{Subject to:} & \mathbb{E}_{\pi}N\leq \rho_{avg},
\end{eqnarray}
where $\rho_{avg}>0$ is a contraint on the average number of relays
(we will describe this procedure in Section~\ref{EN_lambda}).
First, in Sections \ref{nlos_section} to \ref{FPI_NLOS_section},
we work towards obtaining an efficient solution to the problem in (\ref{eq:modified}).

% \begin{lemma}
% \label{relation_lemma}
% Let $\pi_{\lambda}^*\in \Pi$ be an optimal policy for the unconstrained problem in
% (\ref{eq:modified}) such that $\mathbb{E}_{\pi_{\lambda}^*}N=\rho_{avg}$. 
% Then $\pi_{\lambda}^*$ is also optimal for the constrained problem in
% (\ref{eq:main}).
% \end{lemma}

% Sections \ref{nlos_section} to \ref{FPI_NLOS_section} are devoted
% towards obtaining an efficient solution to the problem in
% (\ref{eq:modified}). In Section~\ref{EN_lambda}, before
% reporting our numerical results in Section~\ref{numerical_work_section}, we will briefly (using
% Lemma~\ref{relation_lemma}) describe the method to solve the problem
% in (\ref{eq:main}).

\section{MDP Formulation and Solution} 
\label{nlos_section}
In this section we formulate the problem in (\ref{eq:modified}) as a
total cost infinite horizon MDP and derive the optimal policy in terms
of optimal placement set. We show that this set is characterized by a
two-dimensional boundary, upon crossing which it is optimal to place
a relay.

\subsection{States, Actions, State-Transitions and Cost Structure} 
\label{NLOS_MDP}
We formulate the problem as a sequential decision process starting at
the origin of the lattice path. The decision to place or not place a
relay at the $k$-th step is based on $((M_k,N_k),Z_k)$, where
$(M_k,N_k)$ denotes the coordinates of the deployment person with
respect to the previous relay and $Z_k \in \{\mathsf{e},\mathsf{c}\}$;
$Z_k = \mathsf{e}$ means that at step~$k$ the random lattice path has
ended and $Z_k = \mathsf{c}$ means that the path will continue in the
same direction for at least one more step. Thus, the state space is
given by:
\begin{eqnarray}
\mathcal{S}=\Big\{(m,n,z): (m,n)\in \mathbb{Z}_+^2, z\in \{\mathsf{e},
\mathsf{c}\}\Big\}\cup\{\phi\},
\end{eqnarray}
where $\phi$ denotes the cost-free terminal state, i.e., the state
after the end of the path has been discovered. The action taken at
step $k$ is denoted $U_k\in\{0,1\}$, where $U_k=1$ is the action to
place a relay, and $U_k=0$ is the action of not placing a relay. When
the state is $(m,n,\mathsf{c})$ and when action $u$ is taken, the
transition probabilities are given by:
\begin {itemize}
\item If $u$ is $0$ then,\\
(\textbf{i})  $(m,n,\mathsf{c})$ $\longrightarrow$ $(m+1,n,\mathsf{c})$ w.p.\ $(1-p)q$\\
(\textbf{ii}) $(m,n,\mathsf{c})$ $\longrightarrow$ $(m+1,n,\mathsf{e})$  w.p.\ $pq$\\
(\textbf{iii}) $(m,n,\mathsf{c})$ $\longrightarrow$ $(m,n+1,\mathsf{c})$ w.p.\ $(1-p)(1-q)$\\
(\textbf{iv}) $(m,n,\mathsf{c})$ $\longrightarrow$ $(m,n+1,\mathsf{e})$ w.p.\ $p(1-q)$.

\item If $u$ is $1$ then\\
(\textbf{i}) $(m,n,\mathsf{c})\longrightarrow (1,0,\mathsf{c})$ w.p.\  $(1-p)q$\\
(\textbf{ii}) $(m,n,\mathsf{c})\longrightarrow (1,0,\mathsf{e})$ w.p.\  $pq$\\
(\textbf{iii}) $(m,n,\mathsf{c})\longrightarrow (0,1,\mathsf{c})$ w.p.\  $(1-p)(1-q)$\\
(\textbf{iv}) $(m,n,\mathsf{c})\longrightarrow (0,1,\mathsf{e})$ w.p.\  $p(1-q)$.
\end{itemize}

If $Z_k=\mathsf{e}$ then the only allowable action is $u=1$ and we
enter into the state $\phi$. If the current state is $\phi$, we stay
in the same cost-free termination state irrespective of the control
$u$. The one step cost when the state is $s\in\mathcal{S}$ is given
by:
\begin{eqnarray*}
c(s,u)=\left\{ \begin{array}{l l}
		d(m,n) & \mbox{ if } s=(m,n,\mathsf{e}),\\
		\lambda+d(m,n) & \mbox{ if } u=1 \mbox{ and } s=(m,n,\mathsf{c}), \\
		0 & \mbox{ if } u=0 \mbox{ or } s=\phi . \end{array}\right.		
\end{eqnarray*}
For simplicity we write the state $(m,n,\mathsf{c})$ as simply $(m,n)$.

\subsection{Optimal Placement Set $\mathcal{P}_\lambda$}
\label{nlos_relaxed_problem_section}
Let $J_{\lambda}(m,n)$ denote the optimal cost-to-go when the current
state is $(m,n)$. When at some step the state is $(m,n)$ the
deployment person has to decide whether to place or not place a relay
at the current step.  $J_\lambda$ is the solution of the Bellman
equation \cite[Page~137, Prop.~1.1]{Bertsekas2},
\begin{eqnarray}
J_{\lambda}(m,n)=\min\{c_p(m,n),c_{np}(m,n)\},
\end{eqnarray}
where $c_p(m,n)$ and $c_{np}(m,n)$ denote the expected cost incurred when the decision
is to \emph{place} and \emph{not place} a relay, respectively. $c_p(m,n)$ is given by
\begin{eqnarray} \label{cp}
c_p(m,n)&=&\lambda+d(m,n)+(1-p)(1-q)J_{\lambda}(0,1)\nonumber \\
&&+(1-p)qJ_{\lambda}(1,0)+pd(1).
\end{eqnarray}
The term $\lambda+d(m,n)$ in the above expression is the one step cost which is first 
incurred when a relay is placed. The remaining terms are the average cost-to-go from 
the next step. The term $(1-p)(1-q)J_{\lambda}(0,1)$ can be understood as follows:
$(1-p)(1-q)$ is the probability that the path proceeds Eastward without ending. 
Thus the state at the next step is $(0,1,\mathsf{c})$ w.p.\ $(1-p)(1-q)$, the optimal
cost-to-go from which is, $J_{\lambda}(0,1)$. Similarly for the term $(1-p)qJ_{\lambda}(1,0)$, 
$(1-p)q$ is the probability that the path will proceed, without ending, towards the North
(thus the next state is $(1,0,\mathsf{c})$) and  $J_{\lambda}(1,0)$ is the 
cost-to-go from the next state. Finally, in the term $pd(1)$, $p$ is the probability
that the path will end, either proceeding East or North, at the next step and $d(1)$
is the cost of the last link.
Following a similar explanation, the expression for $c_{np}(m,n)$ can be written as:
\begin{eqnarray} \label{cnp}
\lefteqn{c_{np}(m,n)=}\nonumber\\
&&(1-p)qJ_{\lambda}(m+1,n)+(1-p)(1-q)J_{\lambda}(m,n+1) \nonumber\\
&&+pqd(m+1,n)+p(1-q)d(m,n+1).
\end{eqnarray}

We define the optimal placement set $\mathcal{P}_\lambda$ as the set of all lattice points $(m,n)$, 
where it is optimal to place rather than to not place a relay. Formally,
\begin{eqnarray}
\label{defn}
\mathcal{P}_\lambda=\Big\{(m,n):c_p(m,n)\leq c_{np}(m,n)\Big\}.
\end{eqnarray}
In this definition, if the costs of placing and not-placing are the
same, we have arbitrarily chosen to place at that point.

The above result yields the following main theorem of this section
which characterizes the optimal placement set $\mathcal{P}_\lambda$ in
terms of a boundary.
\begin{theorem} 
\label{placement_boundary}
The optimal placement set $\mathcal{P}_\lambda$ is characterized by a boundary, i.e., there exist 
mappings  $m^*:\mathbb{Z}_+\rightarrow \mathbb{Z}_+$ and  $n^*:\mathbb{Z}_+\rightarrow \mathbb{Z}_+$ such that:
% which define the optimal placement set $\mathcal{P}_\lambda$ as follows:
\begin{eqnarray}
\mathcal{P}_\lambda&=&\bigcup_{n\in\mathbb{Z}_+}\{(m,n): m\geq m^*(n)\}  \label{set11}\\
&=&\bigcup_{m\in\mathbb{Z}_+}\{(m,n): n\geq n^*(m)\}.\label{set22}
\end{eqnarray}
\end{theorem}
\begin{proof}[Proof Outline]
  The proof utilizes the conditions \textbf{C2} and \textbf{C3}
  imposed on the cost function $d(\cdot)$. First, using (\ref{cp}) and
  (\ref{cnp}) in (\ref{defn}) and rearranging we alternatively write
  $\mathcal{P}_\lambda$ as, $\mathcal{P}_\lambda=\{(m,n): F(m,n)\ge
  K\}$, where $K$ is a constant and $F(\cdot,\cdot)$ is some function
  of $m$ and $n$. Then, we complete the proof by showing that $F(m,n)$
  is non-decreasing in both $m$ and $n$. This requires us to prove
  (using an induction argument) that $H_{\lambda}(m,n):=
  J_{\lambda}(m,n)-d(m,n)$ is non-decreasing in $m$ and $n$. Also,
  Lemma~\ref{cor1} has to be used here. For a formal proof see
  Appendix~\ref{placement_boundary_appendix}.
\end{proof}

\emph{Remark:} Though the optimal placement set $\mathcal{P}_\lambda$
was characterized nicely in terms of a boundary $m^*(\cdot)$ and
$n^*(\cdot)$, a naive approach of computing this boundary, using value
iteration to obtain $J_\lambda(m,n)$ (for several values of
$(m,n)\in\mathbb{Z}_+^2$), would be computationally intensive. Our
effort in the next section (Section~\ref{OSLA_formulation_section}) is
towards obtaining an alternate simplified representation for
$\mathcal{P}_\lambda$ using which we propose an algorithm in
Section~\ref{FPI_NLOS_section}, which is guaranteed to return
$\mathcal{P}_\lambda$ in a finite (in practice, small) number of
steps.

\section{Optimal Stopping Formulation} 
\label{OSLA_formulation_section}
We observe that the points where the path has not ended, and a relay is placed, 
are renewal points of the 
decision process. This motivates us to think of the decision process after a relay is placed as an 
optimal stopping problem with \emph{termination cost} $J_\lambda(0,0)$ (which is the optimal cost-to-go 
from a relay placement point). Let $\overline{\mathcal{P}}_\lambda$ denote the placement set 
corresponding to the OSLA rule (to be defined next). In this section we prove our next  main result 
that $\mathcal{P}_\lambda=\overline{\mathcal{P}}_\lambda$.

\subsection{One-Step-Look-Ahead Stopping Set
  $\overline{\mathcal{P}}_{\lambda}$}
Under the OSLA rule, a relay is placed at state $(m,n,\mathsf{c})$ if
and only if the ``cost $c_1(m,n)$ of \emph{stopping} (i.e., placing a
relay) at the current step'' is less than the ``cost $c_2(m,n)$ of
continuing (without placing relay at the current step) for one more
step, and then stopping (i.e., placing a relay at the next step)''. The
expressions for the costs $c_1(m,n)$ and $c_2(m,n)$ can be written as:
\begin{eqnarray*}
c_1(m,n)&=&\lambda+d(m,n)+J_\lambda(0,0)
\end{eqnarray*}
and
\begin{eqnarray*}
\lefteqn{c_2(m,n)=}\nonumber\\
&& pq(d(m+1,n)+p(1-q)d(m,n+1))+(1-p)\nonumber\\
&&\Big(qd(m+1,n)+(1-q)d(m,n+1)+\lambda+J_\lambda(0,0)\Big).\nonumber
\end{eqnarray*}
Then we define the OSLA placement set $\overline{\mathcal{P}}_\lambda$ as:
\begin{eqnarray*}
\overline{\mathcal{P}}_\lambda &=& \{(m,n)\in \mathbb{Z}_{+}^2: c_1(m,n)\le c_2(m,n)\}.  
\end{eqnarray*}
Substituting for $c_1(m,n)$ and $c_2(m,n)$ and simplifying we obtain:
\begin{eqnarray}
\label{OSLA_2}
\overline{\mathcal{P}}_{\lambda}=\Big\{(m,n)\in \mathbb{Z}_{+}^2: p(\lambda+J_{\lambda}(0,0))\leq\Delta_q(m,n)\Big\},
\end{eqnarray}
where $\Delta_q(m,n)=q\Delta_1(m,n)+ (1-q)\Delta_2(m,n)$.
\begin{theorem} 
The OSLA rule is a threshold policy, i.e., there exist mappings  
$\bar{m}:\mathbb{Z}_+\rightarrow \mathbb{Z}_+$ and  $\bar{n}:\mathbb{Z}_+\rightarrow \mathbb{Z}_+$, 
which define the one-step placement set $\overline{\mathcal{P}}_{\lambda}$ as follows, 
\begin{eqnarray}
\overline{\mathcal{P}}_\lambda&=&\bigcup_{n\in\mathbb{Z}_+}\{(m,n): m\geq \bar{m}(n)\}\label{OSLA2_1}\\
&=&\bigcup_{m\in\mathbb{Z}_+}\{(m,n): n\geq \bar{n}(m)\}\label{OSLA2_2}.
\end{eqnarray}
\end{theorem}
\begin{proof}
  Noticing that in (\ref{OSLA_2}) $\Delta_q(m,n)$ is non-decreasing in
  $(m,n)$ and $p(\lambda+J_{\lambda}(0,0))$ is a constant, the proof
  follows along the lines of the proof of
  Theorem~\ref{placement_boundary}.
\end{proof}

Now, we present the main theorem of this section.
\begin{theorem} 
\label{optimality_OSLA_2}
\begin{eqnarray*}
\mathcal{P}_\lambda=\overline{\mathcal{P}}_\lambda.
\end{eqnarray*}
\end{theorem}
\begin{proof}
See Appendix~\ref{optimality_OSLA_2_appendix}.
\end{proof}

\emph{Remark:} The characterization in (\ref{OSLA_2}) is much simpler than the one in (\ref{placement1}) 
once the value of $J_{\lambda}(0,0)$ is given. In the following subsection, we define a function
 $g(\cdot)$ and express 
$J_{\lambda}(0,0)$ as the minimum value of this function. 

\subsection{Computation of $J_{\lambda}(0,0)$}
\label{calculation_cost_section}
Let us start by defining a collection of placement sets indexed by $h \geq 0$:
\begin{eqnarray} \label{NLOS_Placement}
{\mathcal P}(h) = \{(m,n)\in\mathbb{Z}_+^2: p(\lambda + h) \leq\Delta_q(m,n)\}.
\end{eqnarray}
Referring to (\ref{OSLA_2}), note that ${\mathcal P}(J_{\lambda}(0,0))=\overline{\mathcal{P}}_\lambda$. Let
$g(h)$ denote the cost-to-go, starting from $(0,0)$, if the placement set $\mathcal{P}(h)$ is employed. 
Then, since $J_{\lambda}(0,0)$ is the optimal cost-to-go and 
$\mathcal{P}_\lambda \in \{{\mathcal P}(h) \}_{h\geq 0}$, we have $J_{\lambda}(0,0)=\min_{h\geq 0}g(h)$. 

To compute $g(h)$, we proceed by defining the boundary $\mathcal{B}(h)$ of $\mathcal{P}(h)$ as follows:
\begin{eqnarray}
\mathcal{B}(h)&=&\{(m,n)\in \mathcal{P}(h): (m-1,n)\in\mathcal{P}^c(h) \mbox{ or } \nonumber\\
&&\hspace{4mm}(m,n-1) \in \mathcal{P}^c(h) \},
\end{eqnarray}
where $\mathcal{P}^c(h):=\mathbb{Z}_+^2-\mathcal{P}(h)$.

Suppose the corridor ends at some $(m,n)\in \mathcal{P}^c(h)\cup \mathcal{B}(h)$, then only a cost of 
$d(m,n)$ is incurred. Otherwise (i.e., if the corridor reaches some $(m,n)\in \mathcal{B}(h)$ and 
continues), using a renewal argument, a cost of $d(m,n)+\lambda+g(h)$ is incurred, where 
$d(m,n)+\lambda$ is the cost of placing a relay and $g(h)$ is the future cost-to-go. We can thus write:
\begin{eqnarray} 
g(h)&=&\sum_{(m,n)\in \mathcal{P}^c(h)\cup \mathcal{B}(h)}\mathbb{P}((m,n),
\mathsf{e})d(m,n)+\nonumber \\
&&\sum_{(m,n)\in \mathcal{B}(h)}\!\mathbb{P}((m,n),\mathsf{c})(g(h)\!+\!\lambda\!+\!d(m,n)),\;\;\;
\end{eqnarray}
where $\mathbb{P}((m,n),\mathsf{e})$ is the probability of the corridor ending at $(m,n)$ and 
$\mathbb{P}((m,n),\mathsf{c})$ is the probability of the corridor reaching the boundary and continuing. 
Solving for $g(h)$, we obtain:
\begin{eqnarray} \label{NLOS_gh}
\label{Boundary}
g(h)&=&\frac{1}{1-\sum_{(m,n)\in \mathcal{B}(h)}
\mathbb{P}((m,n),\mathsf{c})}\times \nonumber\\
&&\Bigg(\sum_{(m,n)\in \mathcal{P}^c(h)\cup \mathcal{B}(h)}\mathbb{P}((m,n),
\mathsf{e})d(m,n)+\nonumber  \\
&&\sum_{(m,n)\in \mathcal{B}(h)}\mathbb{P}((m,n),\mathsf{c})(\lambda+d(m,n))\Bigg).
\end{eqnarray} 
The above expression is extensively used in our algorithm proposed in the next section.

We conclude this subsection by deriving the expression for the probabilities 
$\mathbb{P}((m,n),\mathsf{e})$ and $\mathbb{P}((m,n),\mathsf{c})$. Let us partition the 
boundary $\mathcal{B}(h)$ into three mutually disjoint sets:
\begin{eqnarray*}
\mathcal{B}^w(h)&=&\{(m,n)\in \mathcal{B}(h): (m-1,n)\in \mathcal{B}(h)\}\\
\mathcal{B}^s(h)&=&\{(m,n)\in \mathcal{B}(h): (m,n-1)\in \mathcal{B}(h)\}\\
\mathcal{B}^{null}(h)&=&\{(m,n)\in \mathcal{B}(h): (m-1,n)\notin \mathcal{B}(h)\mbox{ and }\nonumber\\
&&(m,n-1)\notin \mathcal{B}(h)\}.
\end{eqnarray*}
For a depiction of the various boundary points, see Fig.~\ref{Boundary_figure}.
\begin{figure}[t!]
\centering
\includegraphics[width=0.6\linewidth]{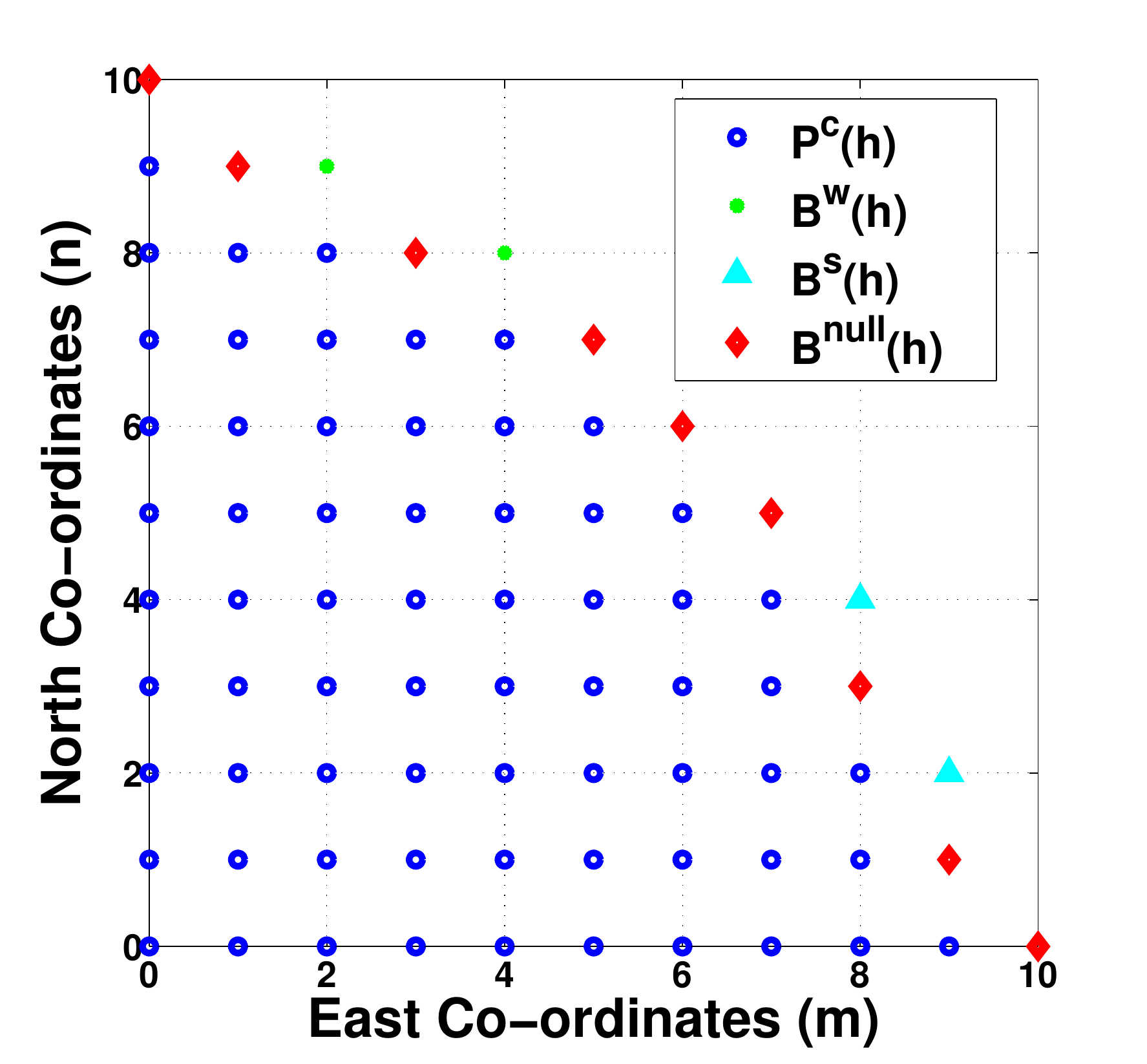}
\caption{Example of placement set of the form in (\ref{NLOS_Placement}): 'o' denotes lattice 
points outside the placement set; lattice points on the boundary can be partitioned into three 
sets according to the direction, from which they can be reached.}
\label{Boundary_figure}
\end{figure}
Now, $\mathbb P((m,n),\mathsf{e})$ can be written as:
\begin{eqnarray*}
\lefteqn{\mathbb P((m,n),\mathsf{e})=}\\
&&\left\{\begin{array}{lr}
\binom{m+n}{m}p(1-p)^{m+n-1}q^m(1-q)^n \\
\hspace{30mm}\mbox{ if } (m,n)\in \mathcal{P}^c(h)\cup \mathcal{B}^{null}(h)\\
\binom{m+n-1}{m} p(1\!-\!p)^{m+n-1} q^m(1\!-\!q)^n \mbox{ if }(m,n)\!\in\! \mathcal{B}^{w}(h)\\
\binom{m+n-1}{m-1} p(1\!-\!p)^{m+n-1} q^m(1\!-\!q)^n \mbox{ if }(m,n)\!\in\! \mathcal{B}^{s}(h).
\end{array}\right.
\end{eqnarray*}
This can be understood as follows. Any point $(m,n)\in \mathcal{P}^c(h)\cup 
\mathcal{B}^{null}(h)$ can be reached from West or South. $\binom{m+n}{m}$ is the number of 
possible paths for reaching $(m,n)$. Each such path has to go $m$ times Eastwards (thus the 
term $q^m$) and $n$ times Northwards (thus the term $(1-q)^n$) and finally ending at $(m,n)$ 
(thus the term $p(1-p)^{m+n-1}$). Any point $(m,n)\in \mathcal{B}^{w}(h)$ can be reached only 
from South point $(m,n-1)$. The probability of reaching $(m,n-1)$ without ending is 
$\binom{m+n-1}{m} (1-p)^{m+n-1} q^m(1-q)^{n-1}$. Then, the corridor reaches $(m,n)$ and ends 
with probability $p(1-q)$. $\mathbb P((m,n),\mathsf{e})$ for $(m,n) \in \mathcal{B}^{s}(h)$ can 
be obtained analogously. 

Similarly, $\mathbb P((m,n),\mathsf{c})$ can be written as:
\begin{eqnarray*}
 \lefteqn{\mathbb P((m,n),\mathsf{c})=}\\
 &&\left\{\begin{array}{ll}
 \binom{m+n}{m}(1-p)^{m+n}q^m(1-q)^n \\
\hspace{30mm}\mbox{ if } (m,n)\in \mathcal{P}^c(h)\cup \mathcal{B}^{null}(h)\\
 \binom{m+n-1}{m} (1-p)^{m+n} q^m(1-q)^n \mbox{ if }(m,n)\in \mathcal{B}^{w}(h)\\
 \binom{m+n-1}{m-1} (1-p)^{m+n} q^m(1-q)^n \mbox{ if }(m,n)\in \mathcal{B}^{s}(h).
\end{array}\right.
 \end{eqnarray*}

\section{OSLA Based Fixed Point Iteration Algorithm} 
\label{FPI_NLOS_section}
In this section, we present an efficient fixed point iteration
algorithm (Algorithm~\ref{Algo}) using the OSLA rule in (\ref{OSLA_2})
for obtaining the optimal placement set, $\mathcal{P}_\lambda$, and
the optimal cost-to-go, $J_\lambda(0,0)$. There are two advantages of
our algorithm over the naive approach of directly trying to minimize
the function $g(\cdot)$ to obtain $J_\lambda(0,0)$ (recall that
$J_\lambda(0,0)=\min_{h\ge0} g(h)$):
\begin{itemize}
\item On the theoretical side, this iterative algorithm avoids
  explicit optimization altogether, which, otherwise would be
  performed numerically over a continuous range. Without any structure
  on the objective function, direct numerical minimization of
  $g(\cdot)$ is difficult and often unsatisfactory, as it  invariably
  uses some sort of heuristic search over this continuous range.
\item On the practical side, this algorithm is proved to converge
  within a finite number of iterations and observed to be extremely
  fast (requires 3 to 4 iterations typically). 
% \textbf{This is highly
%     desirable in applications such as rapid deployment in emergency
%     situations for which the algorithm is targeted. THIS POINT NOT
%     CLEAR SINCE THE POLICY WILL BE COMPUTED OFF-LINE?}
\end{itemize}  

The following is our Algorithm which we refer to as the 
OSLA Based Fixed Point Iteration Algorithm.
\begin{algorithm} 
\caption{OSLA Based Fixed Point Iteration Algorithm}
\begin{algorithmic}[1] 
\REQUIRE $0<p<1$, $0\leq q\leq 1$, $\lambda\geq 0$
\STATE $k=0$, $h^{(k)}=0$
\WHILE{1} 
\STATE $\mathcal{P}(h^{(k)})\gets \{(m,n)\in \mathbb{Z}_{+}^2: p(\lambda+h^{(k)})\leq \Delta_q(m,n)\}$
\STATE Compute $g(h^{(k)})$ using (\ref{Boundary})
\IF{$g(h^{(k)})==h^{(k)}$} \STATE Break; \ENDIF
\STATE $h^{(k+1)}\gets g(h^{(k)})$
\STATE $k\gets k+1$ 
\ENDWHILE
\STATE {\textbf{return} $g(h^{(k)})$, $\mathcal{P}(h^{(k)})$}
\end{algorithmic}
\label{Algo}
\end{algorithm}

We now prove the correctness and finite termination properties of our
algorithm.  First, we define $g^*:=J_{\lambda}(0,0)=\min_{h\geq
  0}g(h)$.  Now consider a sample plot of the function $g(h)$ in
Fig.~\ref{func_g_figure}.  From Fig.~\ref{h_gh} observe that whenever
$h>g^*$ (which is around 150), $h>g(h)$. Also,
Fig.~\ref{h_gh_superimposed} (where we have plotted the functions
$g(h)$ and $l(h)=h$) suggests that $g(h)$ has a unique fixed point. We
formally prove these results.

% \begin{figure}[h]
% \centering
% \includegraphics[width=0.6\linewidth]{../plots/p0pt02qpt5lambda41}
% \caption{Cost-to-go $g(h)$ as a function of $h$ ($p=0.02$, $q=0.5$, $\lambda=41$).}
% \label{h_gh}
% \end{figure}
% \begin{figure}[h!]
% \centering
% \includegraphics[width=0.6\linewidth]{../plots/p0pt02qpt5_line_cross}
% \caption{Zoom on the cost-to-go $g(h)$ as a function of $h$ ($p=0.02$, $q=0.5$, $\lambda=41$).}
% \label{h_gh_superimposed}
% \end{figure}

\begin{figure}[h]
\centering
\subfigure[]{
\includegraphics[width=0.6\linewidth]{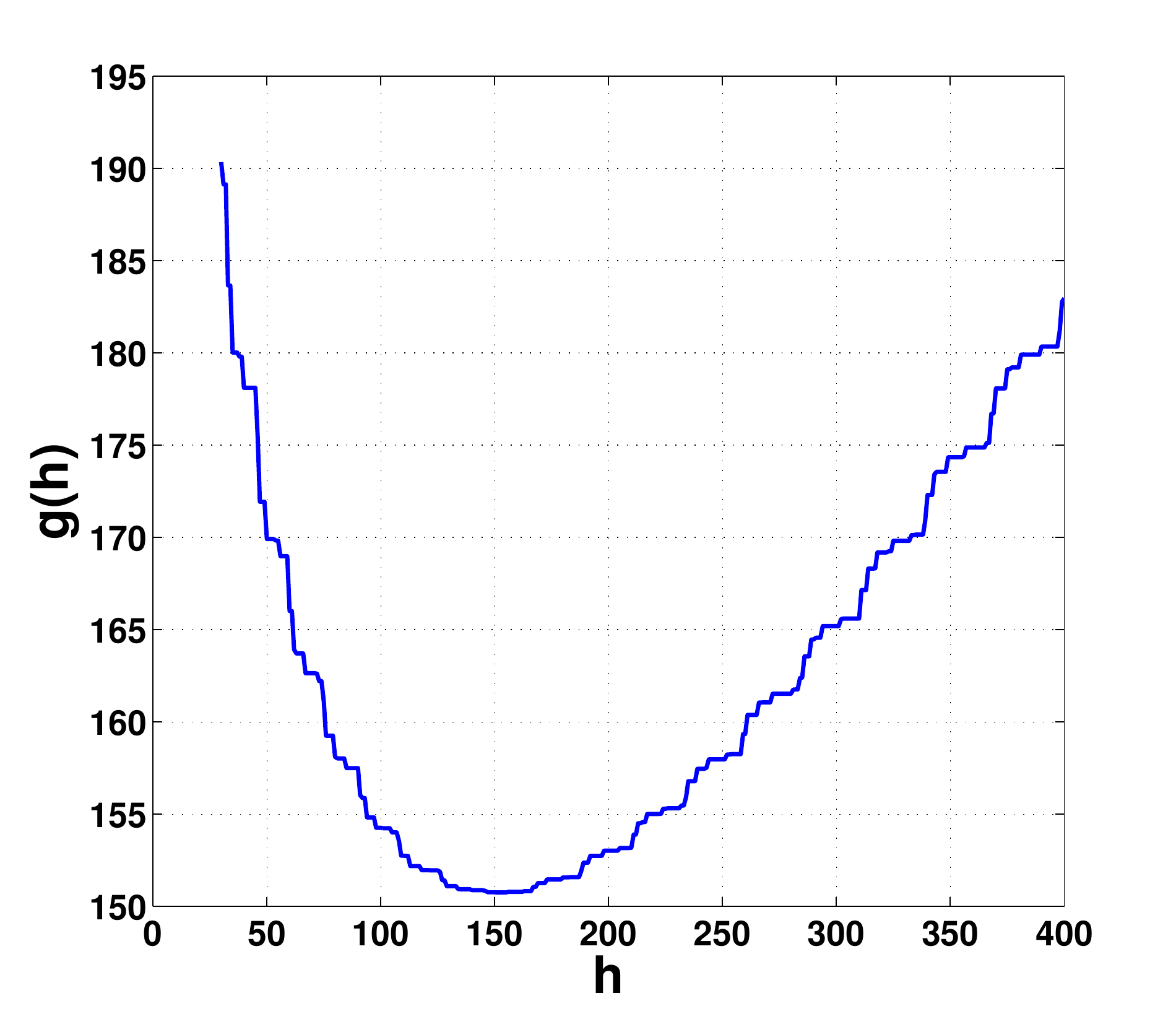}
       \label{h_gh}
}
\subfigure[]{
\includegraphics[width=0.6\linewidth]{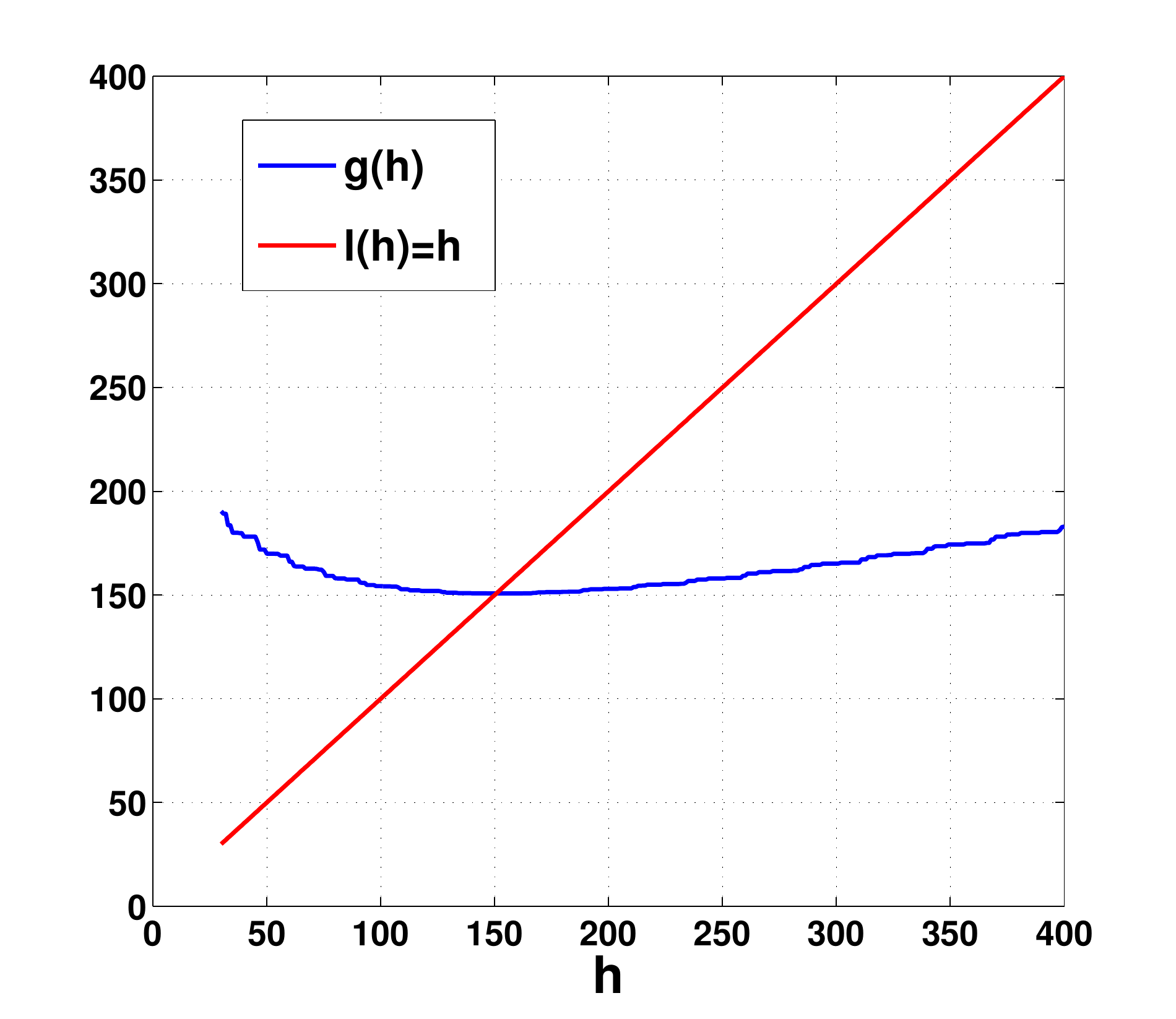}
        \label{h_gh_superimposed}
}
\caption{\subref{h_gh} Cost-to-go $g(h)$ as a function of $h$  \subref{h_gh_superimposed} Zoom on the cost-to-go $g(h)$ as a function of $h$.
These plots are for $p=0.02$, $q=0.5$, and $\lambda=41$.\label{func_g_figure}
}
% \vspace{-6mm}
\end{figure}
 
\begin{lemma} \label{FPI_NLOS}
% Define $g^*:=J_{\lambda}(0,0)=\min_{h\geq 0}g(h)$. 
If $h>g^*$ then $h>g(h)$.
\end{lemma}
\begin{proof}
This follows from the manipulation of (\ref{NLOS_gh}). See Appendix \ref{FPI_NLOS_appendix} 
for details.
\end{proof}
\begin{lemma} \label{uniqueFP_NLOS}
 $g(h)$ has a unique fixed point.
\end{lemma}
\begin{proof}
From (\ref{NLOS_Placement}) and (\ref{OSLA_2}), we observe that 
$\mathcal{P}(J_{\lambda}(0,0))=\overline{\mathcal{P}}_{\lambda}$. 
From Theorem \ref{optimality_OSLA_2}, $\overline{\mathcal{P}}_{\lambda}$ is the optimal 
placement set and thus the cost-to-go of using $\mathcal{P}(J_{\lambda}(0,0))$ is $J_{\lambda}(0,0)$, 
i.e., $g(J_{\lambda}(0,0))=J_{\lambda}(0,0)$. Hence, $J_{\lambda}(0,0)=g^*$ is a fixed 
point of $g(\cdot)$. Now, any $h>g^*$ cannot be a fixed point since, in this case, $h>g(h)$ 
from Lemma~\ref{FPI_NLOS}. On the other hand, any $h<g^*$ is such that $h<g^*\leq g(h)$ because 
$g^*$ is the optimal cost-to-go. Hence, $g^*$ is the unique fixed point of $g(\cdot)$. 
\end{proof}

We are now ready to prove the convergence property of our Algorithm.

\begin{lemma}\label{hk_nonincreasing_NLOS}
\begin{enumerate}
\item The sequence $\{h^{(k)}\}_{k\geq 1}$ (in Algorithm~\ref{Algo}) is non-increasing,
i.e., $h^{(k+1)}\leq h^{(k)}$, with the equality sign holding if and only if $h^{(k)}=g^*$.
\item The sequence ${\{\mathcal{P}^c(h^{(k)})}\}_{k\geq 1}$ is non-increasing, i.e., $\mathcal{P}^c(h^{(k+1)})\subseteq \mathcal{P}^c(h^{(k)})$, where the containment is strict whenever $\mathcal{P}^c(h^{(k+1)})\varsubsetneq{\mathcal{P}_{\lambda}}^c$.
\end{enumerate}
\end{lemma}
\begin{proof}
1) Note first that $h^{(k)}\geq g^*$ for $k\geq 1$ because $h^{(k)}=g(h^{(k-1)})\geq g^*$. Then, for $k \geq 1$, we have either $h^{(k)} = g^*$ or  $h^{(k)} > g^*$. In the first case $h^{(k+1)} = g(h^{(k)})=g(g^*)=g^*=h^{(k)}$ and we can stop, whereas in the second case, from Lemma \ref{FPI_NLOS} we have  $h^{(k+1)} = g(h^{(k)}) < h^{(k)}$. 

2) From (\ref{NLOS_Placement}), $h_2>h_1$ implies $\mathcal{P}^c(h_1)\subseteq  \mathcal{P}^c(h_2)$. 
Hence, as $\{h^{(k)}\}_{k\geq 1}$ is non-increasing (from Part~1)), 
${\{\mathcal{P}^c(h^{(k)})}\}_{k\geq 1}$ is also non-increasing. 

Suppose $\mathcal{P}^c(h^{(k+1)})=\mathcal{P}^c(h^{(k)})$ then $g(h^{(k+1)})=g(h^{(k)})=h^{(k+1)}$
(second equality is by the definition of $\{h^{(k)}\}$), which implies $h^{(k+1)}=g^*$ (since $g(\cdot)$ has a unique fixed point, see Lemma~\ref{uniqueFP_NLOS}). Thus, $\mathcal{P}^c(h^{(k+1)})={\mathcal{P}_{\lambda}}^c$.
% From (\ref{NLOS_Placement}) it is clear that $\mathcal{P}^c(h^{(k)})={\mathcal{P}_{\lambda}}^c$ 
% whenever $h^{(k)}=g^*$. Conversely, suppose $\mathcal{P}^c(h^{(k)})={\mathcal{P}_{\lambda}}^c$ then
% $g(h^{(k)})=g^*$, which implies $g(h^{(k)})=h^{(k)}$ (since $g(\cdot)$ has a unique fixed point,
% see Lemma~\ref{uniqueFP_NLOS}). Thus, we have 
% $\mathcal{P}^c(h^{(k)})={\mathcal{P}_{\lambda}}^c$ if and only if $h^{(k)}=g^*$.
% Finally, to conclude the proof, note that $h^{(k)}=g^*$ if and only if 
% $h^{(k+1)}=g(h^{(k)})=h^{(k)}$, so that $\mathcal{P}^c(h^{(k+1)})=\mathcal{P}^c(h^{(k)})$.
% $\mathcal{P}^c(h^{(k)})={\mathcal{P}_{\lambda}}^c$ if and only if $h^{(k)}=g^*$. Then, $h^{(k+1)}=g(h^{(k)})=g(g^*)=g^*=h^{(k)}$. Thus, $\mathcal{P}^c(h^{(k+1)})=\mathcal{P}^c(h^{(k)})$.
\end{proof}

\begin{theorem}
Algorithm~\ref{Algo} returns $g^*$ and ${\mathcal{P}_{\lambda}}^c$ in a finite number of steps. 
\end{theorem}
\begin{proof}
Noting that $h^{(1)}=g(h^{(0)})\geq g^*$ and using (\ref{NLOS_Placement}), we have 
${\mathcal{P}_{\lambda}}^c \subseteq \mathcal{P}^c(h^{(1)})$. Either ${\mathcal{P}_{\lambda}}^c = \mathcal{P}^c(h^{(1)})$, in which case the algorithm stops. Otherwise, note that both sets, ${\mathcal{P}_{\lambda}}^c$ and $\mathcal{P}^c(h^{(1)})$ contain a finite number of lattice points (from the definition of $\mathcal{P}(h)$ in (\ref{NLOS_Placement})). Using Lemma~\ref{hk_nonincreasing_NLOS}, $\mathcal{P}^c(h^{(k)})$ converges to ${\mathcal{P}_{\lambda}}^c$ in at most $|\mathcal{P}^c(h^{(1)})-{\mathcal{P}_{\lambda}}^c|<\infty$ iterations. Once $\mathcal{P}^c(h^{(k)})$ converges to ${\mathcal{P}_{\lambda}}$, the algorithm stops and returns the optimal cost-to-go $g^*$.
\end{proof}

\section{Solving the Constrained Problem} 
\label{EN_lambda}
In this section, we devise a method to solve the constrained problem in (\ref{eq:main}) using the 
solution of the unconstrained problem (\ref{eq:modified}) provided by Algorithm~\ref{Algo}. 
This method is applied in Section~\ref{distance_heuristic_section} where, imposing a constraint
on the average number of relays, we compare the performance of
a distance based heuristic with the optimal.

We begin with the following standard result which relates the solutions
of the problems in (\ref{eq:modified}) and (\ref{eq:main}).
\begin{lemma}
\label{relation_lemma}
Let $\pi_{\lambda}^*\in \Pi$ be an optimal policy for the unconstrained problem in
(\ref{eq:modified}) such that $\mathbb{E}_{\pi_{\lambda}^*}N=\rho_{avg}$. 
Then $\pi_{\lambda}^*$ is also optimal for the constrained problem in
(\ref{eq:main}).
\end{lemma}
However, the above lemma is useful only when we are able to exhibit a $\lambda$ such that
$\mathbb{E}_{\pi_{\lambda}^*}N=\rho_{avg}$. The subsequent development in this section
is towards obtaining the solution to the more general case.

The expected number of relays used by the optimal policy, $\pi_{\lambda}^*$, which uses the optimal 
placement set $\mathcal{P}_{\lambda}$, can be computed as:
\begin{eqnarray}
\mathbb{E}_{\pi_{\lambda}^*}N= \frac{\sum_{(m,n)\in \mathcal{B}_{\lambda}}\mathbb{P}((m,n),\mathsf{c})}{1-\sum_{(m,n)\in \mathcal{B}_{\lambda}}\mathbb{P}((m,n),\mathsf{c})},
\end{eqnarray}
where $\mathbb{P}((m,n),\mathsf{c})$ is the reaching probability corresponding to $\mathcal{P}_\lambda$ 
and $\mathcal{B}_{\lambda}$ is the boundary of $\mathcal{P}_\lambda$.
A plot of $\mathbb{E}_{\pi_{\lambda}^*}N$ vs.\ $\lambda$ is given in Fig.~\ref{EN_EC_lambda_figure}. 
We make the following observations about $\mathbb{E}_{\pi_{\lambda}^*} N$.

1) $\mathbb{E}_{\pi_{\lambda}^*} N$ decreases with $\lambda$; this is as expected, since as 
each relay becomes ``costlier'' fewer relays are used on the average.

2) Even when $\lambda = 0$, $\mathbb{E}_{\pi_{\lambda}^*} N$ is finite. This is because 
$d(0)> 0$, i.e., there is a positive cost for a $0$ length link. Define the value of 
$\mathbb{E}_{\pi_{\lambda}^*} N$  with $\lambda = 0$ to be $\rho_{\max}$.

3) $\mathbb{E}_{\pi_{\lambda}^*} N$ vs. $\lambda$ is a piecewise constant function. This occurs 
because the relay placement positions are discrete. For a range of values of $\lambda$ the 
same threshold is optimal. This structure is also evident from the results based on the 
optimal stopping formulation and the OSLA rule in Section~\ref{OSLA_formulation_section}. It 
follows that for a value of $\lambda$ at which there is a step in the plot, there are two 
optimal deterministic policies, $\underline{\pi}$ and $\overline{\pi}$, for the relaxed 
problem. Let $\underline{\rho}=\mathbb{E}_{\underline{\pi}} N$ and 
$\overline{\rho}=\mathbb{E}_{\overline{\pi}} N$.

We have the following structure of the optimal policy for the constrained problem:
\begin{theorem}
\label{constrained_solution_theorem}
\begin{enumerate}
	\item For $\rho_{avg} \geq \rho_{max}$ the optimal placement set is obtained for $\lambda=0$,
 i.e., is $\mathcal{P}_0$.
	 \item For $\rho_{avg} < \rho_{max}$, if there is a $\lambda$ such that (a) 
$\mathbb{E}_{\pi_{\lambda}^*} N=\rho_{avg}$ then the optimal policy is $\pi_{\lambda}^*$, or 
(b) $\underline{\rho}<\rho_{avg}<\overline{\rho}$ then the optimal policy is obtained by 
mixing $\underline{\pi}$ and $\overline{\pi}$.
\end{enumerate}
\end{theorem}
\begin{proof}
  1) is straight forward. For proof of 2)-(a), see
  Lemma~\ref{relation_lemma}. Considering now 2)-(b), define
  $0<\alpha<1$ such that
  $(1-\alpha)\underline{\rho}+\alpha\bar{\rho}=\rho_{avg}$.  We obtain
  a mixing policy $\pi_{m}$ by choosing $\underline{\pi}$ w.p.\
  $1-\alpha$ and $\bar{\pi}$ w.p.\ $\alpha$ at the beginning of the
  deployment. For any policy $\pi$ we have the following standard
  argument:
\begin{eqnarray}
\lefteqn{\mathbb{E}_{\pi_m}C+\lambda\mathbb{E}_{\pi_m}N}\nonumber\\
&=&(1-\alpha)
(\mathbb{E}_{\underline{\pi}}C+\lambda\underline{\rho})+\alpha(\mathbb{E}_{\bar{\pi}}C
+\lambda\bar{\rho})\nonumber\\
&\le& (1-\alpha)
(\mathbb{E}_{\pi}C+\lambda\mathbb{E}_{\pi}N)+\alpha(\mathbb{E}_{{\pi}}C+\lambda
\mathbb{E}_{\pi}N)\nonumber\\
&=&\mathbb{E}_{{\pi}}C+\lambda\mathbb{E}_{\pi}N.
\end{eqnarray}
The inequality is because $\underline{\pi}$ and $\overline{\pi}$ are both optimal for the 
problem (\ref{eq:modified}) with relay price $\lambda$.
Thus, we have shown that $\pi_m$ is also optimal for the relaxed problem. Using this along 
with $\mathbb{E}_{\pi_m}N=\rho_{avg}$ in Lemma~\ref{relation_lemma}, we conclude the proof.
\end{proof}

\section{Numerical Work}
\label{numerical_work_section}
% As mentioned earlier, before reporting our numerical results we first describe the method to solve
% the constrained problem in (\ref{eq:main}). 
For our numerical work we use the one-hop power function
$d(r)=P_m+\gamma r^{\eta}$, with $P_m=0.1$, $\gamma=0.01$.
 We first study the effect of parameter variation
on the various costs. Next, we compare the performance of 
a distance based heuristic with the optimal.

\begin{figure}[t!]
\centering
\includegraphics[width=0.6\linewidth]{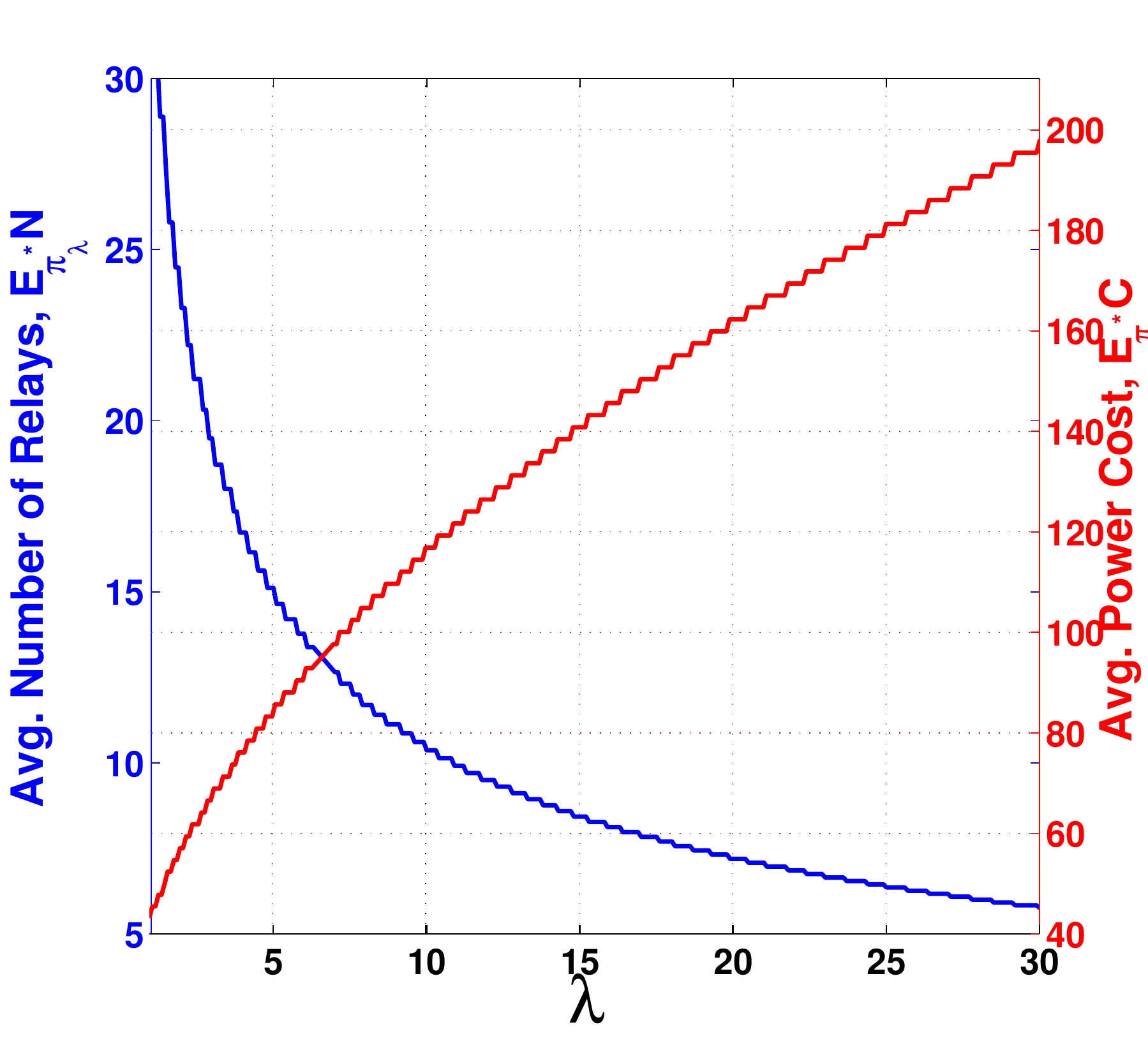}
\caption{Average number of relays $\mathbb{E}_{\pi_{\lambda}^*} N$ (left) and average power cost $\mathbb{E}_{\pi_{\lambda}^*} C$ (right) as a function of $\lambda$ ($p=0.002$, $q=0.5$ and $\eta=2$).}
\label{EN_EC_lambda_figure}
% \vspace{-2mm}
\end{figure}

\begin{figure}[t]
\centering
\includegraphics[width=0.6\linewidth]{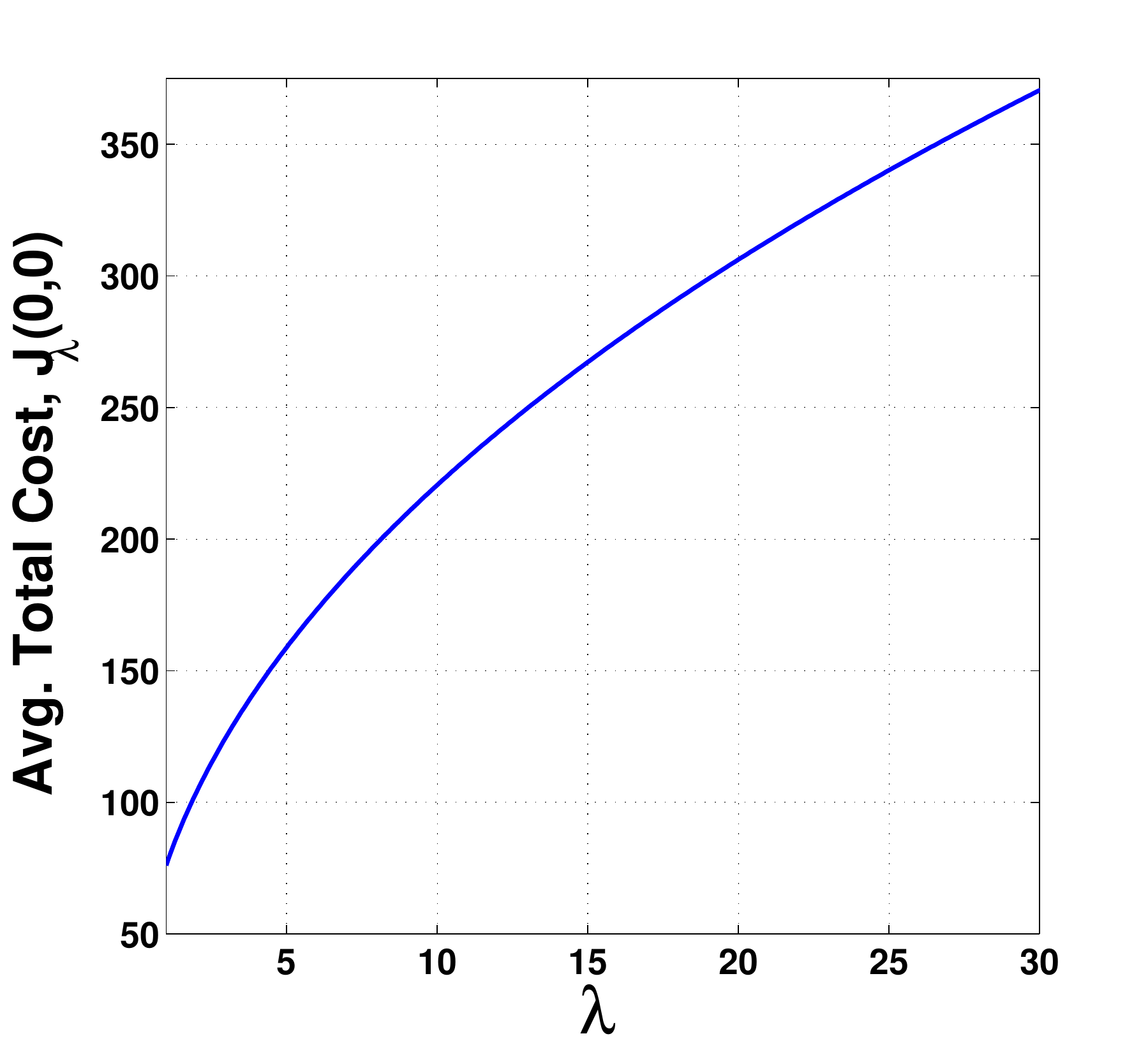}
\caption{Average total cost $J_{\lambda}(0,0)$ as a function of $\lambda$ ($p=0.002$, $q=0.5$ and $\eta=2$).}
\label{J00_lambda_figure}
% \vspace{-6mm}
\end{figure}

\subsection{Effect of Parameter Variation}
In Fig.~\ref{Boundary_figure}, we have already shown an optimal placement
boundary for $p = 0.002$, $q = 0.5$, and $\eta=3$. Since $q=0.5$ the
boundary is symmetric about the $m=n$ line.

In Fig.~\ref{EN_EC_lambda_figure}, we plot $E_{\pi_\lambda^*}N$ and
$E_{\pi_\lambda^*}C$ vs.\ $\lambda$. The plot of $J_\lambda(0,0)$ vs.\
$\lambda$ is in Fig.~\ref{J00_lambda_figure}.  These plots are for
$p=0.002$ and $q=0.5$. Since $\lambda$ is the cost per relay, as
expected, $E_{\pi_\lambda^*}N$ decreases as $\lambda$ increases. We
observe that $E_{\pi_\lambda^*} C$ and the optimal total cost
$J_\lambda(0,0)$ increase as $\lambda$ increases. A close examination
of Fig.~\ref{EN_EC_lambda_figure} reveals that both the plots are step
functions.  This is due to the discrete placement at lattice points,
which results in the same placement boundary being optimal for a range
of $\lambda$ values. Thus, as seen in Section~\ref{EN_lambda}, at the
$\lambda$ values, where there is jump in $E_{\pi_\lambda^*} N$, a
random mixture of two policies is needed.

Fig.~\ref{J00_q_figure} shows the variation of the total optimal cost
$J_\lambda(0,0)$ with $q$. The variation is symmetric about
$q=0.5$. For a given probability $p$ of the path ending, $q = 0.5$
results in the path folding frequently. In such a case, since NLOS
propagation is permitted, and the path-loss is isotropic, fewer
relays are required to be placed. On the other hand, when $q$ is close
to $0$ or to $1$ the path takes fewer turns and more relays are
needed, leading to larger values of the total cost.

In Fig.~\ref{boundary_variation_figure} we show the variation of optimal boundaries with 
$\eta$. As $\eta,$ the path-loss exponent, increases the hop cost increases for a given hop 
distance. This results in relays needing to be placed more frequently. As can be seen the 
placement boundaries shrink with increasing $\eta$. We also notice that the placement boundary 
for $\eta = 2$ is a straight line; indeed this provable result holds for $\eta=2$ for any 
values of $p$ and $q$. 

\begin{figure}[t!]
\centering
\includegraphics[width=0.6\linewidth]{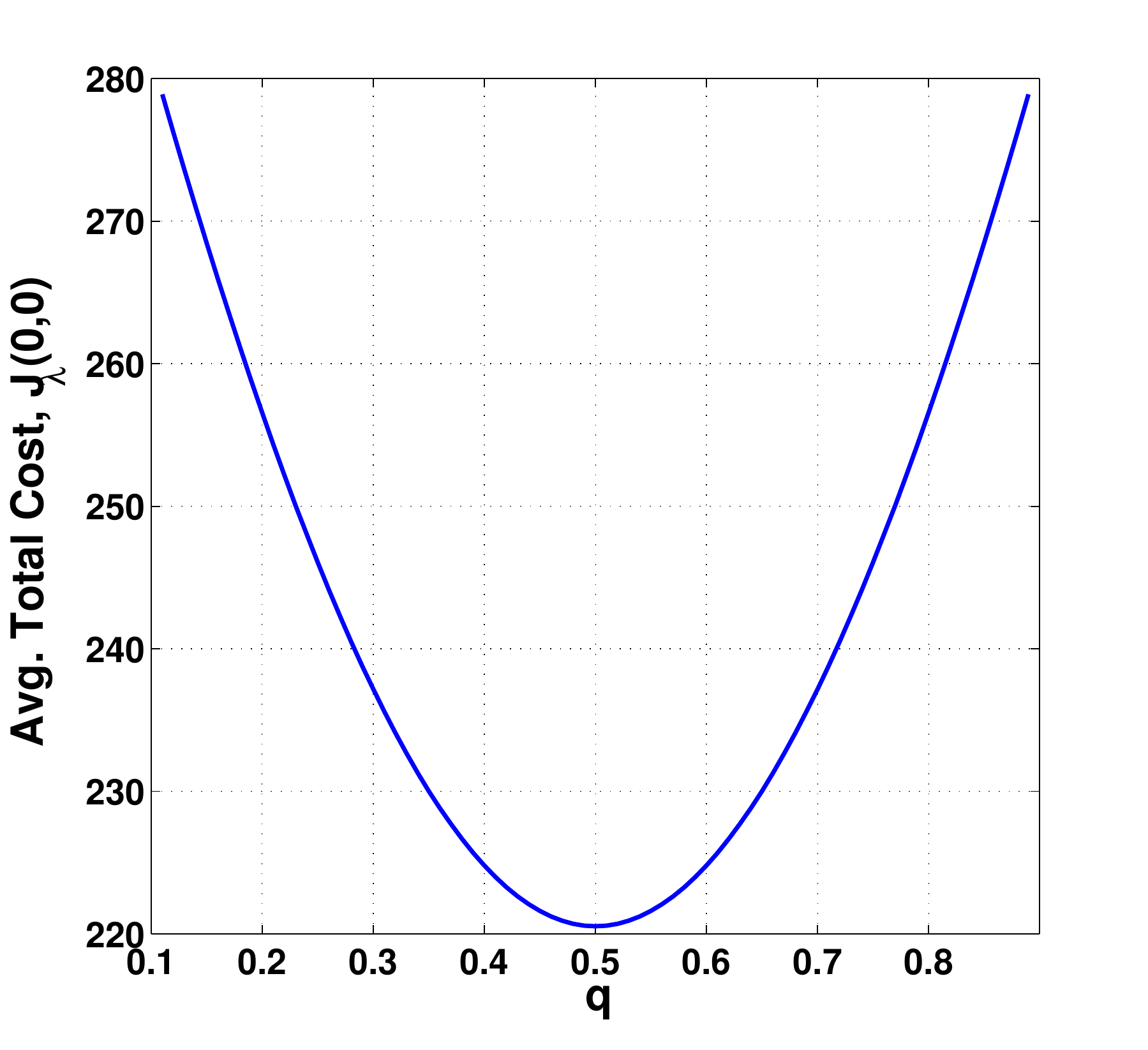}
\caption{Average total cost $J_{\lambda}(0,0)$ as a function of $q$ ($p=0.002$ and $\eta=2$).}
\label{J00_q_figure}
% \vspace{-4mm}
\end{figure}

\begin{figure}[t]
\centering
\includegraphics[width=0.6\linewidth]{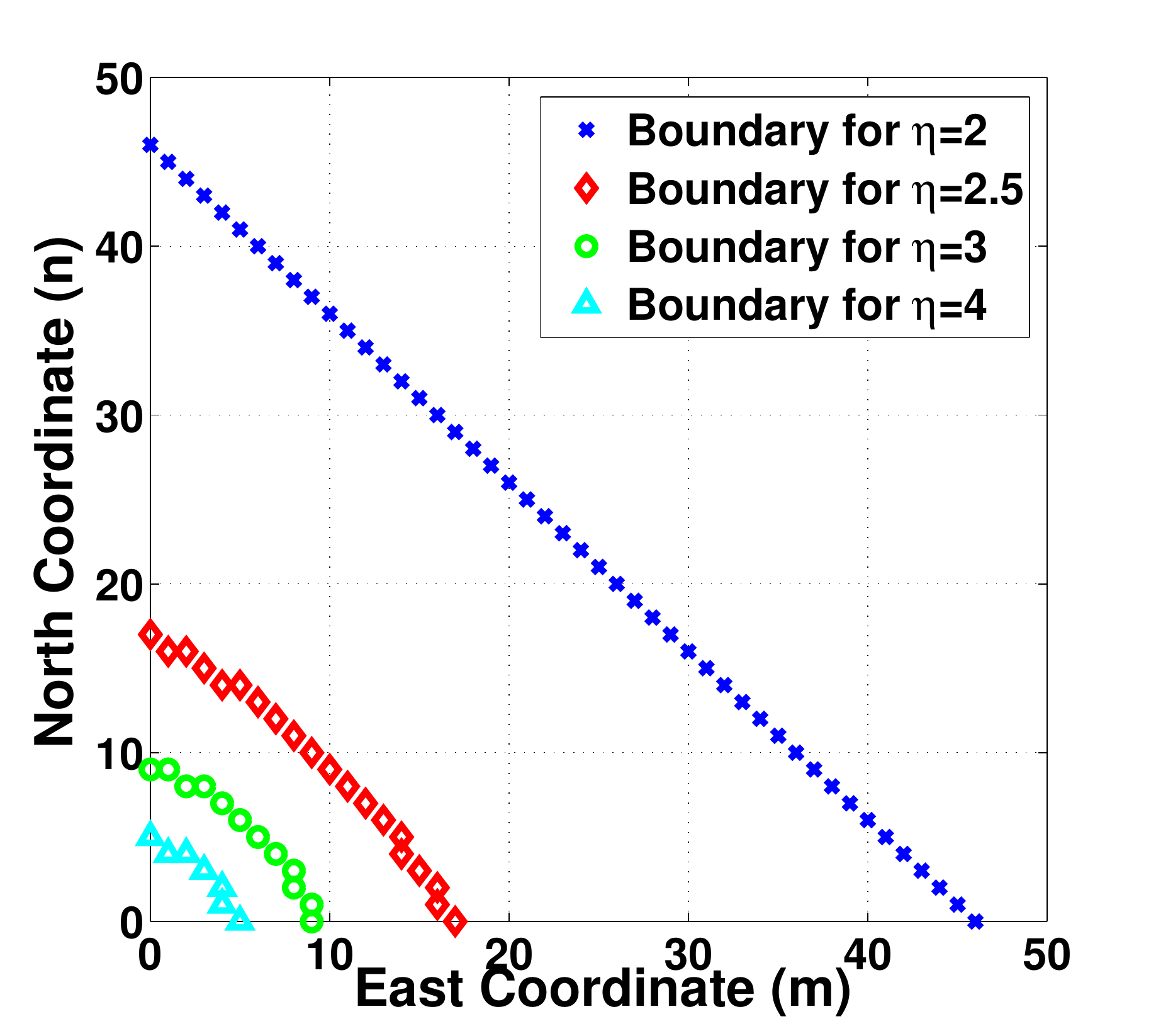}
\caption{Boundaries for various values of the path-loss exponent $\eta$ ($p=0.002$, $q=0.5$).}
\label{boundary_variation_figure}
% \vspace{-6mm}
\end{figure}

\subsection{Comparison with the Distance based Heuristic}
\label{distance_heuristic_section}
We recall from the literature survey in Section~\ref{intro} that prior work invariably proposed 
the policy of placing a relay after the RF signal strength from the previous relay dropped 
below a threshold. For isotropic propagation (as we have assumed in this paper), this is 
equivalent to placing the relay after a circular boundary is crossed. With this in mind, we 
obtained the \emph{optimal constant distance placement policy} (called \emph{the heuristic} 
hereafter) numerically in a manner similar to what is described in 
Section~\ref{calculation_cost_section}. A sample result is provided in 
Fig.~\ref{boundary_comparisons_figure}, for the parameters $p=0.002$, $q=0.5$ and $\eta=2$. We 
observe that if the path were to evolve roughly Eastward or Northward then the heuristic will 
result in many more relays being placed. On the other hand, if the path evolves diagonally 
(which has higher probability) then the two placement boundaries will result in similar 
placement decisions.

\begin{figure}[t!]
\centering
\includegraphics[width=0.6\linewidth]{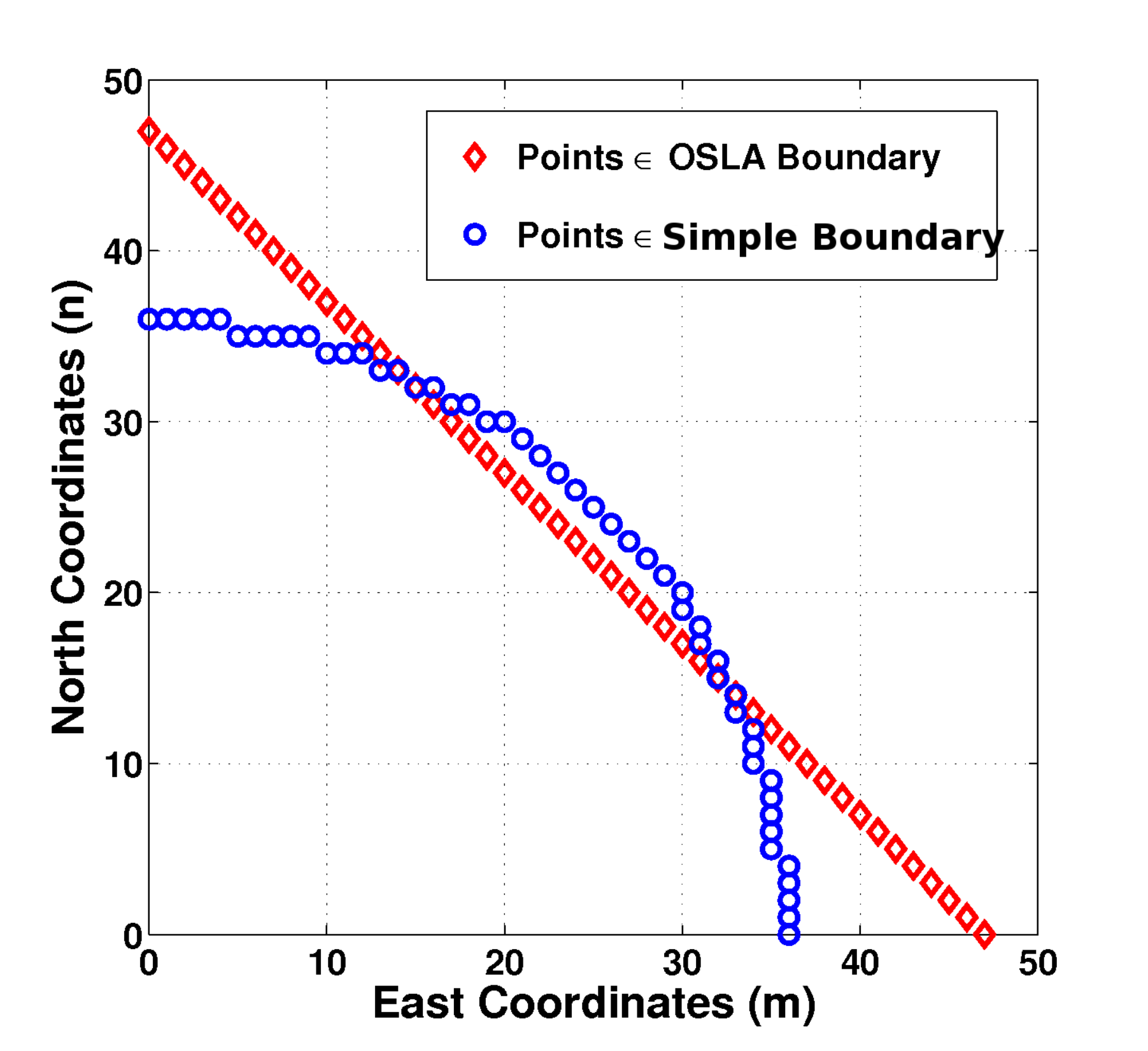}
\caption{Boundary of the optimal placement set (OSLA boundary) and boundary derived from the 
heuristic policy ($p=0.002$, $q=0.5$ and $\eta=2$). }
\label{boundary_comparisons_figure}
\vspace{-4mm}
\end{figure}

\begin{figure}[t]
\centering
\includegraphics[width=0.6\linewidth]{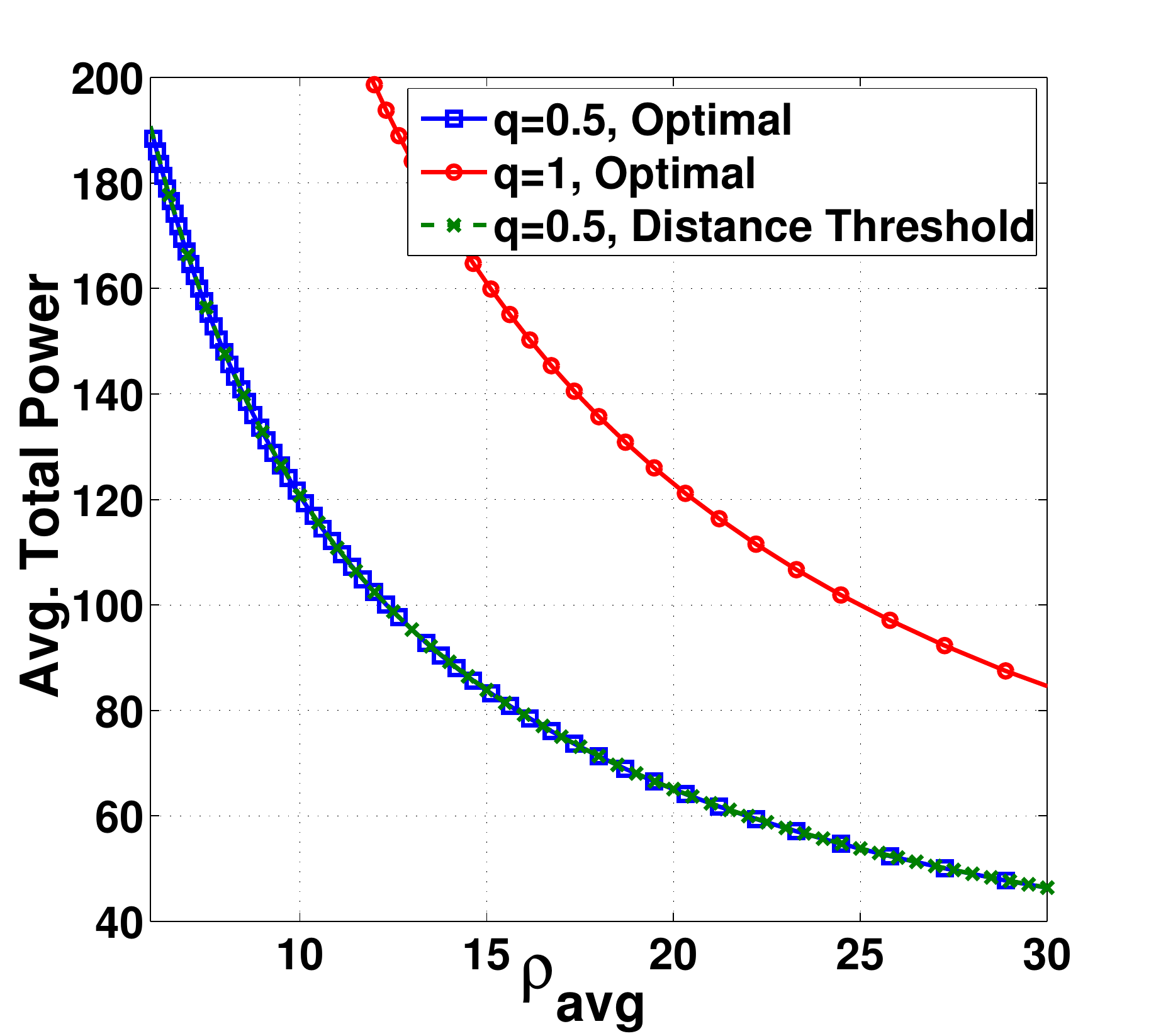}
\caption{Average total power as a function of $\rho$ for the optimal policy ($q=0.5$ and $q=1$, 
which corresponds to the straight line) and for the heuristic ($q=0.5$) for $p=0.002$ and 
$\eta=2$.}
\label{EC_rho_comparison_figure}
\vspace{-6mm}
\end{figure}

This observation shows up in Fig.~\ref{EC_rho_comparison_figure}, where we show the cost 
incurred by the optimal policy (for $q=0.5$ and for $q=1$, which corresponds to a straight line 
corridor) and the heuristic ($q=0.5$) vs.\ $\rho$ for the constrained problem. As expected, the 
cost is much larger for $q=1$ since the path does not fold. We find that for $q=0.5$ the 
optimal placement boundary and the heuristic provide costs that are almost indistinguishable at 
this scale. We have performed simulations by varying the system parameters and observed the 
same good performance of the optimal constant distance placement policy. This suggests that the 
heuristic policy performs well provided that the threshold distance is optimally chosen with 
respect to the system parameters. 

\section{Conclusion}
We considered the problem of placing relays on a random lattice path
to optimize a linear combination of average power cost and average
number of relays deployed. The optimal placement policy was proved to
be of threshold nature (Theorem~\ref{placement_boundary}).  We further
proved the optimality of the OSLA rule (in
Theorem~\ref{optimality_OSLA_2}).  We have also devised an OLSA based
fixed point iteration algorithm (Algorithm~\ref{Algo}), which we have
proved to converge to the optimal placement set in a finite number of
steps.  Through numerical work we observed that the performance (in
terms of average power incurred for a given relay constraint) of the
optimal policy is closed to that of the distance threshold policy
provided that the threshold distance is optimally chosen with respect
to the system parameters.

\bibliographystyle{IEEEtran}
\bibliography{IEEEabrv,BibFile-Infocom}

\clearpage
\appendices

\section{Proof of Lemmas in Section~\ref{system_model_section}}
\subsection{Proof of Lemma \ref{conv}}
\label{conv_appendix}
\begin{proof}
Any norm is convex so that the function $g(x,y)\equiv \sqrt{x^2+y^2}$ is convex 
in $(x,y)$. The delay function $d(\cdot)$ is also assumed to be convex and non-decreasing in its argument. Hence by using the composition rule 
\cite[Section~3.2.4]{Boyd}, we conclude that the function $d(x,y)\equiv d(\sqrt{x^2+y^2})$ is convex in $(x,y)\in \mathbb{R}^2$. 
\end{proof}

\subsection{Proof of Lemma \ref{cor1}}
\label{cor1_appendix}
\begin{proof}
It is easier to prove the lemma allowing the arguments $m$ and $n$ take values from the Real line. 
We have,
\begin{eqnarray*}
\Delta_1(x,y)=d(x+\delta,y)-d(x,y)
\end{eqnarray*}
Partially differentiating both sides w.r.t. $x$, we get
\begin{eqnarray*}
\frac{\partial\Delta_1(x,y)}{\partial x}&=& d_x(x+\delta,y)-d_x(x,y)\\
&=&\delta d_{xx}(\zeta,y) \mbox{ where } x<\zeta<x+\delta\label{LMVT1}\\
&>&0, \label{positivity}
\end{eqnarray*}
where the equality follows from the application of Lagrange's Mean Value Theorem to the function $d_x(.,y)$ and the inequality is due to assumption in (\ref{assumption}).
The above proves the fact that $\Delta_1(x,y)$ is non-decreasing in $x$. 

To prove that $\Delta_1(x,y)$ is non-decreasing in $y$, we partially differentiate $\Delta_1(x,y)$ 
w.r.t.\ $y$ and obtain
\begin{eqnarray*}
\frac{\partial\Delta_1(x,y)}{\partial y}&=&d_y(x+\delta,y)-d_y(x,y)\\
&=&\delta d_{xy}(\eta,y) \mbox{ where } x<\eta<x+\delta\label{LMVT2}\\
&>&0, \label{pos2}
\end{eqnarray*} 
where the equality  follows from the application of Lagrange's Mean Value Theorem to the function $d_y(.,y)$ and the inequality is due to assumption in (\ref{assumption}).
This shows that the function $\Delta_1(x,y)$ is non-decreasing in both the coordinates $x$ and $y$. In a similar way it can also be shown that $\Delta_2(x,y)$ is non-decreasing in $x$ and $y$ under the assumption made in (\ref{assumption}). This completes the proof.
\end{proof}

\section{Proof of Theorem~\ref{placement_boundary}}
\label{placement_boundary_appendix}
%%%%%%%%%%%%%%%%%%%%%%%%%%%%%%%%%%%%%%%%%%%%%%%%%%%%%%%%%%%%%%%%%%%%%%%%%%%
We begin by defining $H_{\lambda}(m,n):= J_{\lambda}(m,n)-d(m,n)$. Substituting for 
$c_p(m,n)$ and $c_{np}(m,n)$ 
(from (\ref{cp}) and (\ref{cnp}), respectively) into (\ref{defn}) and rearranging we obtain 
(recall the definitions of $\Delta_1(m,n)$ and $\Delta_2(m,n)$ from Section~\ref{system_model_section}):
\begin{eqnarray} \label{placement1}
\lefteqn{\mathcal{P}_\lambda=}\nonumber\\
&&\Big\{(m,n)\!:\!(1\!-\!p)(qH_\lambda(m\!+\!1,n)\!+\!(1\!-\!q)H_\lambda(m,n\!+\!1))\nonumber\\
&&+p(q\Delta_1(m,n)+(1-q)\Delta_2(m,n))\geq \lambda+ \\
&&(1-p)qJ_{\lambda}(1,0)+(1-p)(1-q)J_{\lambda}(0,1)+pd(1)\Big\}. \nonumber
\end{eqnarray}
\begin{lemma} 
\label{H}
For a fixed $\lambda$, $H_\lambda(m,n)$ is non-decreasing in both $m\in\mathbb{Z}_+$ and $n\in 
\mathbb{Z}_+$.
\end{lemma}
\begin{proof}
Consider a sequential relay placement problem where we have $K$ steps to go. The corridor length is the minimum of $K$ and of a geometric random variable with parameter $p$. The problem be formulated as a finite horizon MDP with horizon length $K$. For any given $(m,n)$, $J_{K}(m,n)$, $K\geq 2$ is obtained recursively: 
\begin{eqnarray*}
\lefteqn{J_{K}(m,n) = \min \{c_p(m,n),c_{np}(m,n)\}}\\
&=&\min\{\lambda+d(m,n)+(1-p)qJ_{K-1}(1,0)+pqd(1)+ \\
&& (1-p)(1-q)J_{K-1}(0,1)+p(1-q)d(1),\\
&& (1-p)qJ_{K-1}(m+1,n)+pqd(m+1,n)+\\
&& (1\!-\!p)(1\!-\!q)J_{K\!-\!1}(m,n\!+\!1)\!+\!p(1\!-\!q)d(m,n\!+\!1)\}.
\end{eqnarray*}
For $K=1$, since a sensor must be placed at the next step, we have 
$J_{1}(m,n)=\min\{\lambda+d(m,n)+d(1),qd(m+1,n)+(1-q)d(m,n+1)\}.$
Therefore, 
\begin{eqnarray*}
\lefteqn{H_{1}(m,n):=J_{1}(m,n)-d(m,n)}\\
%&=&\min\{\lambda+d(1),q(d(m+1,n)-d(m,n))+\\
%&&(1-q)(d(m,n+1)-d(m,n))\} \\
&=&\min \{\lambda+d(1),q\Delta_1(m,n)+(1-q)\Delta_2(m,n)\}.
\end{eqnarray*}
From Lemma ~\ref{cor1}, it follows that $H_1(m,n)$ is non-decreasing in both $m$ and $n$.
Now we make the induction hypothesis and assume that $H_{K-1}(m,n)$ is non-decreasing in $m$ and $n$. We have:
\begin{eqnarray*}
\lefteqn{H_{K}(m,n)=J_{K}(m,n)-d(m,n)}\\
% &=&\min\{\lambda+(1-p)qJ_{K-1}(1,0)+pqd(1)\\
% &&+(1-p)(1-q)J_{K-1}(0,1)+p(1-q)d(1),\\
% &&(1-p)(qH_{K-1}(m+1,n)+(1-q)H_{K-1}(m,n+1))+\\
% &&q(d(m+1,n)-d(m,n))+(1-q)(d(m,n+1)-d(m,n))\}\\
&=&\min\{\lambda+(1-p)qJ_{K-1}(1,0)+pqd(1)+\\
&&(1-p)(1-q)J_{K-1}(0,1)+p(1-q)d(1),(1-p)\\
&&(qH_{K-1}(m+1,n)+(1-q)H_{K-1}(m,n+1))+\\
&&q\Delta_1(m,n)+(1-q)\Delta_2(m,n)\}.
\end{eqnarray*}
By the induction hypothesis and Lemma ~\ref{cor1}, it follows that $H_{K}(m,n)$ is non-decreasing in both $m$ and $n$. The proof is complete by taking the limit as $ K\rightarrow \infty$. 
\end{proof}

% \subsection{Proof of Theorem \ref{placement_boundary}}
We are now ready to prove Theorem~\ref{placement_boundary}.
\begin{proof}[Proof of Theorem~\ref{placement_boundary}]
Referring to (\ref{placement1}), utilizing Lemma~\ref{H} and the Lemma ~\ref{cor1}, it follows that for a fixed $n\in \mathbb{Z}_+$, the LHS (Left Hand Side) of (\ref{placement1}), describing the placement set $\mathcal{P}_\lambda$ is an increasing function of $m$, while the RHS (Right Hand Side) is a \emph{finite} constant. Also, because of the assumed properties of the function $d(.)$, $\Delta_1(m,n)\rightarrow \infty$ as $m\rightarrow \infty$, for any 
fixed $n$. Hence it follows that there exists an $m^*(n)\in\mathbb{Z}_+$ such that 
$(m,n)\in \mathcal{P}_\lambda\hspace{10pt} \forall m\geq m^*(n)$. Hence we may write 
%\begin{eqnarray}
$P_\lambda=\bigcup_{n\in\mathbb{Z}_+}\{(m,n)|m\geq m^*(n)\}.$
%\end{eqnarray}
The second characterization follows by similar arguments. 
\end{proof}

\section{Proof of Theorem \ref{optimality_OSLA_2}}
\label{optimality_OSLA_2_appendix}

We require the following lemmas to prove Theorem \ref{optimality_OSLA_2}.

\begin{lemma}
$\mathcal{P}_\lambda \subset \overline{\mathcal{P}}_{\lambda}$
\end{lemma}

\begin{proof}
Suppose that $(m,n)\in \mathcal{P}_\lambda$. 
Then from (\ref{set11}) $(m+1,n)\in \mathcal{P}_\lambda$ and from 
(\ref{set22}), $(m,n+1)\in \mathcal{P}_\lambda$. 
Since $(m,n)\in \mathcal{P}_\lambda$, we have from  (\ref{cp}), (\ref{cnp}) and (\ref{defn}) that 
\begin{eqnarray} 
\label{rs}
\lefteqn{\lambda\!+\!d(m,n)\!+\!(1\!-\!p)qJ_{\lambda}(1,0)\!+\!pqd(1)\!+\!(1\!-\!p)(1\!-\!q)\times} \nonumber\\
&&J_{\lambda}(0,1)\!+\!p(1\!-\!q)d(1)\!\leq\!
(1\!-\!p)qJ_{\lambda}(m\!+\!1,n)\!+pq\!\times\nonumber\\
&&d(m\!\!+\!\!1,n)\!\!+\!\!(1\!\!-\!\!p)(1\!\!-\!\!q)J_{\lambda}(m,n\!\!+\!\!1)\!\!+\!\!p(1\!\!-\!\!q)d(m,n\!\!+\!\!1).\nonumber\\
\end{eqnarray}
Also we may argue that at the state $(0,0)$, it is optimal not to place. Indeed, if it had been optimal to place at the state $(0,0)$, at the next step, we return to the same state, viz., $(0,0)$. Now, because of the stationarity of the optimal policy, we would keep placing relays at the same point, and since ``relay-cost'' $\lambda>0$ and $d(0,0)>0$, the expected cost for this policy would be $\infty$. Hence,
\begin{eqnarray}\label{J00} 
\lefteqn{J_{\lambda}(0,0)=(1-p)qJ_{\lambda}(1,0)+pqd(1)+} \nonumber \\
&&(1-p)(1-q)J_{\lambda}(0,1)+p(1-q)d(1).
\end{eqnarray} 
Since $(m+1,n)\in \mathcal{P}_\lambda$ and $(m,n+1)\in \mathcal{P}_\lambda$, we have (noticing that it is optimal to place at these points and utilizing (\ref{cp}) and (\ref{J00})),  
\begin{eqnarray} 
J_{\lambda}(m+1,n)&=&\lambda+d(m+1,n)+J_{\lambda}(0,0) \label{rs1}\\
J_{\lambda}(m,n+1)&=&\lambda+d(m,n+1)+J_{\lambda}(0,0)\label{rs2}.
\end{eqnarray}
Now, using (\ref{J00}), (\ref{rs1}) and (\ref{rs2}) in (\ref{rs}), we obtain:
\begin{eqnarray}
p(\lambda\!+\!J_{\lambda}(0,0))&\!\leq\!& q\Delta_1(m,n)\!+\!(1\!-\!q)\Delta_2(m,n).
\end{eqnarray}
This proves that 
$(m,n)\in \bar{\mathcal{P}}_\lambda$ and hence
$\mathcal{P}_\lambda\subset \overline{\mathcal{P}}_{\lambda}$
\end{proof}

Using the above Lemma and from (\ref{set11}), (\ref{set22}),
(\ref{OSLA2_1}), (\ref{OSLA2_2}) we can conclude that:
% % Taking into account the threshold structure of the sets (viz Eqns (\ref{set11}), (\ref{set22}), 
% (\ref{OSLA2_1}), (\ref{OSLA2_2})) the above lemma implies that
\begin{eqnarray} 
&n^*(m)\geq \overline{n}(m) & \forall m\in \mathbb{Z}_+ \label{ineq1}\\
&m^*(n)\geq \overline{m}(n) & \forall n\in \mathbb{Z}_+ \label{ineq2}.
\end{eqnarray}

\begin{lemma}
\label{lem33}
If $(m,n)\in\overline{\mathcal{P}}_\lambda$ is such that $(m,n+1)\in\mathcal{P}_\lambda$ and 
$(m+1,n)\in \mathcal{P}_\lambda$, then $(m,n)\in\mathcal{P}_\lambda$
\end{lemma}

\begin{proof}
Since $(m,n)\in \bar{\mathcal{P}}_{\lambda}$, we have from (\ref{OSLA_2}),
\begin{eqnarray} \label{ineq}
p(\lambda+J_{\lambda}(0,0))\leq q\Delta_1(m,n)+(1-q)\Delta_2(m,n).
\end{eqnarray}
Now $(m,n+1)\in \mathcal{P}_\lambda$, and $(m+1,n)\in \mathcal{P}_\lambda$, hence we have from (\ref{rs1}) and (\ref{rs2}):
\begin{eqnarray*}
J_{\lambda}(m+1,n)&=&\lambda+d(m+1,n)+J_{\lambda}(0,0)\\
J_{\lambda}(m,n+1)&=&\lambda+d(m,n+1)+J_{\lambda}(0,0).
\end{eqnarray*}
The expression (\ref{J00}) is always true. Now using (\ref{J00}) and the above two equations in inequality (\ref{ineq}), we obtain (\ref{rs}), which proves that $(m,n)\in \mathcal{P}_\lambda$.
\end{proof}

\begin{lemma} \label{lemma:mplusk}
If $(m,n)\in \mathcal{P}_\lambda$ (resp. $\overline{\mathcal{P}}_\lambda$), then $(m+k,n)\in \mathcal{P}_\lambda$ (resp. $\overline{\mathcal{P}}_\lambda$) and $(m,n+k)\in \mathcal{P}_\lambda$ (resp. $\overline{\mathcal{P}}_\lambda$) for any $k\in \mathbb{Z}_+$.
\end{lemma}
\begin{proof}
The proof follows easily because the LHS of (\ref{placement1}) is increasing in both $m$ and $n$ while the RHS is a constant. Similarly, the RHS of (\ref{OSLA_2}) is increasing in both $m$ and $n$ while the LHS is a constant. 
\end{proof}

We can now prove the main theorem. 

\begin{proof}[Proof of Theorem \ref{optimality_OSLA_2}]
We need to show that inequalities in (\ref{ineq1}) and (\ref{ineq2}) are equalities. For any $m\in \mathbb{Z}_+$, suppose that in (\ref{ineq1}) $n^*(m)> n^*(m)-1\geq \bar{n}(m)$. Then we have the following inclusions: 
\begin{eqnarray}
(m,n^*(m)) &\in& \mathcal{P}_\lambda \nonumber \\
(m,n^*(m)-1) &\in& \overline{\mathcal{P}}_\lambda \nonumber \\
(m,n^*(m)-1) &\notin& \mathcal{P}_\lambda. \label{eq55}
\end{eqnarray}
Let us index the collection of lattice-points $(m+i,n^*(m)-\!1)$ by $N_i, i\in \mathbb{Z}_+$. Since $(m,n^*(m)-1) \in \overline{\mathcal{P}}_\lambda$, from Lemma~\ref{lemma:mplusk}, it follows that $N_i\in \overline{\mathcal{P}}_\lambda$. From (\ref{eq55}), $N_0\notin \mathcal{P}_\lambda$.

Then, the optimal policy being a threshold policy, we know that there exists a finite $k>0$, s.t. $N_k\in \mathcal{P}_\lambda$, i.e.,
\begin{equation} \label{eq:Nk}
(m+k,n^*(m)-1) \in \mathcal{P}_\lambda.
\end{equation}
Again from Lemma~\ref{lemma:mplusk}, since $(m,n^*(m))\in \mathcal{P}_\lambda$, we have for any $k>0$:
\begin{eqnarray} \label{eq:Nk2}
(m+k-1,n^*(m)) &\in& \mathcal{P}_\lambda.
\end{eqnarray}
Now we see that for the point $N_{k-1}$, the conditions of Lemma~\ref{lem33} are satisfied. Hence $N_{k-1}\in \mathcal{P}_\lambda$. If $k=1$, we already have a contradiction since $N_0\notin \mathcal{P}_\lambda$. Otherwise for $k>1$, using Lemma~\ref{lemma:mplusk} and $N_{k-1}\in \mathcal{P}_\lambda$, we can show that 
$N_{k-2}$ is subject to the conditions of Lemma~\ref{lem33} implying that $N_{k-2}\in \mathcal{P}_\lambda$. By iteration, we finally obtain that $N_0\in \mathcal{P}_\lambda$, which contradicts (\ref{eq55}) and proves the result. 
\end{proof}

\section{Proof of Lemma~\ref{FPI_NLOS}}
\label{FPI_NLOS_appendix}

We start by showing the following lemma. 
\begin{lemma} \label{g_h_Eqn}
For \emph{any} placement set $\mathcal{P}(h)$ of the form in (\ref{NLOS_Placement}), we have:
\begin{eqnarray} \label{mainEqn_2D}
\sum_{(m,n)\in\mathcal{P}^c(h)}r(m,n)\bigg(\Delta_q(m,n)-p(\lambda+g(h)) \bigg)&& \nonumber \\
+d(0,0)+\lambda=0, &&
\end{eqnarray}
where $r(m,n)=(1-p)^{m+n}\binom{m+n}{m}q^m(1-q)^n$.
\end{lemma}
\begin{proof}
We first introduce some notations and definitions.

Let us define a path $\sigma$ as a possible realization of the corridor, starting from $(0,0)$ and let $\mathbb{P}(\sigma)$ be the probability of such a path. The set of all paths is denoted by $\Sigma$. Let $\Sigma_{mn}$ denote the set of all paths that end at $(m,n)\in \mathcal{P}^c(h)\cup \mathcal{B}(h)$ and $\Sigma_{mn}(c)$ the set of all paths that hit $(m,n)\in \mathcal{B}(h)$ and continue. 

%is a path through the integer lattice, starting from the origin $(0,0)$. If the point $(m,n)$ is on the path $\sigma$, the path $\sigma$ either stops at that point (w.p. $p$) or continues at least one step further. In the latter case, the next point on the path is $(m+1,n)$ w.p. $q$ or, $(m,n+1)$ w.p. $(1-q)$. The set of all paths is denoted by $\Sigma$. The set of all paths that end at the point $(m,n)$ is denoted by $\Sigma_{mn}$, $(m,n)\in\mathcal{P}^c(h)\bigcup \mathcal{B}(h)$. The set of paths that continue beyond the boundary $\mathcal{B}(h)$ is denoted by $\Sigma(c)=\Sigma-\bigcup_{(m,n)\in\mathcal{P}^c(h)\bigcup \mathcal{B}(h)}\Sigma_{mn}=\bigcup_{(m,n)\in\mathcal{B}(h)}\Sigma_{mn}(c)$. Where $\Sigma_{mn}(c)$ denotes the set of paths that first hit $\mathcal{B}(h)$ at the point $(m,n)\in \mathcal{B}(h)$ and continue. Consider a path $\sigma \in \Sigma$. It is completely characterized by its edge set ${E}_\sigma$.
Let us denote the set of edges whose both end vertices belong to the set $\mathcal{P}^c(h)\cup \mathcal{B}(h)$ by $E$. A path $\sigma$ is completely characterized by its edge set ${E}_\sigma$.

The reaching probability, $r(m,n)$, of a point $(m,n)$ is defined as the probability that a random path $\sigma$ \emph{reaches} the point $(m,n)$ and continues for at least one step. Hence, $r(m,n)=(1-p)^{m+n}\binom{m+n}{m}q^m(1-q)^n$.

The incremental cost function  $\delta : E\longrightarrow \mathbb{R}_+$ is defined as follows:
\begin{eqnarray}
\delta(e) = \begin{cases} d(m+1,n)-d(m,n)=\Delta_1(m,n) & \\ \hspace{2cm}\mbox{if } e=\{(m,n),(m+1,n)\} \\
d(m,n+1)-d(m,n)=\Delta_2(m,n) & \\ \hspace{2cm}\mbox{if } e=\{(m,n),(m,n+1)\}. \end{cases} 
\end{eqnarray}

For $(m,n) \in \sigma$, the incremental cost function allows us to write:
\begin{eqnarray}
 d(m,n)=\sum_{e\in E_\sigma \cap E}\delta(e) + d(0,0).
\end{eqnarray}
Now consider 
\begin{eqnarray}
 &&\hspace{-1cm}\sum_{\mathcal{P}^c(h)\cup \mathcal{B}(h)}\mathbb{P}((m,n),
\mathsf{e})d(m,n)+ \sum_{\mathcal{B}(h)}\mathbb{P}((m,n),\mathsf{c})d(m,n)\nonumber \\
&=& \sum_{\mathcal{P}^c(h)\cup \mathcal{B}(h)}\sum_{\sigma\in\Sigma_{mn}}\mathbb{P}(\sigma)\bigg(\sum_{e\in E_\sigma}\delta(e)+d(0,0)\bigg) + \nonumber \\
&&\sum_{\mathcal{B}(h)}\sum_{\sigma \in\Sigma_{mn}(c)}\mathbb{P}(\sigma)
\bigg(\sum_{e\in E_\sigma \cap E}\delta(e)+d(0,0)\bigg)\nonumber \\
&=&\sum_{e\in E} \delta(e) \sum_{\sigma \in \Sigma : e\in E_\sigma}\mathbb{P}(\sigma) + d(0,0)\nonumber \\
&=&\sum_{e\in E} \delta(e) t(e) + d(0,0), \label{2D_sum}
\end{eqnarray}
where by $t(e)$ we denote the probability that a random path goes through the edge $e\in E$.

Now if $e$ is horizontal, i.e., $e=\{(m,n),(m+1,n)\}, (m,n)\in\mathcal{P}^c(h)$, we have $t(e)=q r(m,n)$ and $\delta(e)=\Delta_1(m,n)$. Similarly if $e$ is vertical, i.e., $e=\{(m,n),(m,n+1)\}, (m,n)\in\mathcal{P}^c(h)$, we have $t(e)=(1-q) r(m,n)$ and $\delta(e)=\Delta_2(m,n)$. Using these relations, we may rewrite (\ref{2D_sum}) as follows:
\begin{eqnarray}
&& \hspace{-1cm}\sum_{\mathcal{P}^c(h)}r(m,n)\bigg(q\Delta_1(m,n)+(1-q)\Delta_2(m,n)\bigg)+d(0,0) \nonumber\\
&=&\sum_{\mathcal{P}^c(h)}r(m,n)\Delta_q(m,n)+d(0,0).\label{2D_sum_part1}
\end{eqnarray}

Now consider the probability $\sum_{(m,n)\in \mathcal{B}(h)}\mathbb{P}((m,n),\mathsf{c})$. It is the probability that a random path continues beyond the boundary $\mathcal{B}(h)$. Hence we may write
\begin{eqnarray}
\sum_{\mathcal{B}(h)}\mathbb{P}((m,n),\mathsf{c})&=&1-\sum_{\mathcal{P}^c(h)\cup\mathcal{B}(h)}\mathbb{P}((m,n),\mathsf{e})\nonumber\\
&=&1-\sum_{\mathcal{P}^c(h)}r(m,n)p. \label{2D_sum_part2}
\end{eqnarray}
Using (\ref{2D_sum_part1}) and (\ref{2D_sum_part2}) in (\ref{NLOS_gh}) and simplifying, we obtain the result. 
\end{proof}

\begin{proof}[Proof of Lemma~\ref{FPI_NLOS}]

We recall the definition of $\mathcal{P}^c(h)$. 
\begin{eqnarray} \label{defforPc}
\mathcal{P}^c(h)=\{(m,n)\in\mathbb{Z}_+^2: p(\lambda+h)>\Delta_q(m,n)\}.
\end{eqnarray}
Since $h>g^*$, we immediately conclude that ${\mathcal{P}_{\lambda}}^c\subset \mathcal{P}^c(h)$. From (\ref{mainEqn_2D}) in Lemma~\ref{g_h_Eqn}, we may write for the optimal placement set $\mathcal{P}_{\lambda}$:
\begin{eqnarray}\label{EqnP*}
 \sum_{{\mathcal{P}_{\lambda}}^c}r(m,n)\Delta_q(m,n)&=& p(\lambda+g^*)\sum_{{\mathcal{P}_{\lambda}}^c}r(m,n)\nonumber \\
 &&-(d(0,0)+\lambda).
\end{eqnarray}
We may similarly write for the placement set $\mathcal{P}(h)$:
\begin{eqnarray} \label{EqnPh}
 \sum_{{\mathcal{P}}^c(h)}r(m,n)\Delta_q(m,n)&=&p(\lambda+g(h))\sum_{{\mathcal{P}}^c(h)}r(m,n)\nonumber \\
 &&- (d(0,0)+\lambda).
\end{eqnarray}
Now, since ${\mathcal{P}_{\lambda}}^c\subset \mathcal{P}^c(h)$, we may expand the LHS of (\ref{EqnPh}) as follows:
\begin{eqnarray}
\lefteqn{\sum_{{\mathcal{P}}^c(h)}r(m,n)\Delta_q(m,n)}\nonumber \\
&=& \sum_{{\mathcal{P}}_{\lambda}^c}r(m,n)\Delta_q(m,n)+\!\!\!\sum_{{\mathcal{P}}^c(h)\backslash {\mathcal{P}}_{\lambda}^c}r(m,n)\Delta_q(m,n) \nonumber \\
&<& \sum_{{\mathcal{P}}_{\lambda}^c}r(m,n)\Delta_q(m,n)+p(\lambda+h)\!\!\!\sum_{{\mathcal{P}}^c(h)\backslash {\mathcal{P}}_{\lambda}^c}r(m,n) \nonumber \\
&=& p(\lambda+g^*)\sum_{{\mathcal{P}}_{\lambda}^c}r(m,n) - (d(0,0)+\lambda) \nonumber \\
&& +\;p(\lambda+h)\sum_{{\mathcal{P}}^c(h)\backslash {\mathcal{P}}_{\lambda}^c}r(m,n), \label{subforP*}
\end{eqnarray}
%\begin{eqnarray}
% %&&\sum_{(m,n)\in{\mathcal{P}}^c(h)}r(m,n)\Delta(m,n)\nonumber \\
%&& \sum_{{\mathcal{P}^*}^c}r(m,n)\Delta_q(m,n) +\sum_{{\mathcal{P}^c}(h)\backslash{{\mathcal{P}^*}^c}}r(m,n)\Delta_q(m,n) \nonumber \\
%&<&\sum_{{\mathcal{P}^*}^c}r(m,n)\Delta_q(m,n) + p(\lambda+h)\sum_{{\mathcal{P}^c}(h)\backslash{{\mathcal{P}^*}^c}}r(m,n) \label{propforPc}\\
%&=& p(\lambda+g^*)(\sum_{{\mathcal{P}^*}^c}r(m,n)) -(d(0,0)+\lambda)+\nonumber \\
%&&p(\lambda+h)\times\sum_{{\mathcal{P}^c}(h)\backslash{{\mathcal{P}^*}^c}}r(m,n) \label{subforP*}
%\end{eqnarray}
where, for the inequality, we used (\ref{defforPc}) and for (\ref{subforP*}), we have substituted the value for the quantity from (\ref{EqnP*}).
We may alternatively write the RHS of (\ref{EqnPh}) as:
\begin{eqnarray}
\lefteqn{p(\lambda+g(h))\sum_{{\mathcal{P}}^c(h)}r(m,n) -(d(0,0)+\lambda)}\nonumber \\
&=&p(\lambda+g(h))\bigg(\sum_{{\mathcal{P}_{\lambda}}^c}r(m,n)+ \sum_{{\mathcal{P}}^c(h)\backslash{\mathcal{P}_{\lambda}}^c}r(m,n)\bigg)\nonumber\\ 
&&-\;(d(0,0)+\lambda).\label{cmpr2}
\end{eqnarray}
Now comparing (\ref{subforP*}) and (\ref{cmpr2}) and rearranging, we may write:
%\begin{eqnarray*}
%&& p(\lambda+g(h))\bigg(\sum_{{\mathcal{P}^*}^c}r(m,n)+\sum_{{\mathcal{P}}^c(h)\backslash{\mathcal{P}^*}^c}r(m,n)\bigg)<\nonumber \\
% && p(\lambda+g^*)\bigg(\sum_{{\mathcal{P}^*}^c}r(m,n)\bigg)+p(\lambda+h)\sum_{{\mathcal{P}^c}(h)\backslash{{\mathcal{P}^*}^c}}r(m,n)
%\end{eqnarray*}
%Rearrenging the above inequality, we get 
\begin{eqnarray}\label{ineq_h_gh}
 p(g(h)\!-\!g^*)\sum_{{\mathcal{P}_{\lambda}}^c}r(m,n)<p(h\!-\!g(h))\!\!\!
\sum_{{\mathcal{P}^c}(h)\backslash{{\mathcal{P}_{\lambda}}^c}}r(m,n)\!
\end{eqnarray}
Now $\sum_{{\mathcal{P}^c}(h)\backslash{{\mathcal{P}_{\lambda}}^c}}r(m,n)=0$ if and only if ${\mathcal{P}^c}(h)\backslash{{\mathcal{P}_{\lambda}}^c}= \varnothing $, i.e., ${\mathcal{P}}(h)={\mathcal{P}_{\lambda}}$. In this case we get $g(h)=g^*<h$. 
On the other hand, if $\sum_{{\mathcal{P}^c}(h)\backslash{{\mathcal{P}_{\lambda}}^c}}r(m,n)>0$, since $g^*\leq g(h)$, from the inequality (\ref{ineq_h_gh}), we conclude that $h>g(h)$.
\end{proof}

\end{document}